\documentclass[12pt]{article}

\usepackage{amsmath}
\usepackage[english]{babel}
\usepackage[ansinew]{inputenc}
\usepackage{epsfig}
\usepackage{color,graphicx}
\usepackage{wrapfig}
\usepackage{amssymb,amsthm,amsmath}
\usepackage{dsfont}
\usepackage{natbib}
\usepackage{abstract}
\usepackage{mathrsfs}

\usepackage{enumerate}
\usepackage{multirow}
\usepackage{bm}
\usepackage{mathtools}
\usepackage{mathabx}
\usepackage{amsfonts}
\usepackage{bbm}
\usepackage{accents}

\mathtoolsset{showonlyrefs=true}

\usepackage{tabularx}
\usepackage{booktabs}
\usepackage{tikz}
\usepackage{multirow}
\usepackage[skip=-1ex]{caption}
\usepackage{ragged2e}
\usepackage{dsfont}
\usepackage{relsize}
\usepackage{lipsum}
\usepackage[usestackEOL]{stackengine}
\usepackage[toc,page]{appendix}
\usepackage{listings}
\usepackage{graphicx,psfrag,epsf, subfig}
\usepackage{pgffor}
\usepackage{float}
\usepackage{graphics}
\usepackage{fancyhdr}

\graphicspath{{Images/}}

\newtheorem{thm}{Theorem}

\newtheorem{lemma}{Lemma}
\newtheorem{defn}{Definition}
\newtheorem{prop}{Proposition}

\definecolor{mycolor}{RGB}{30,75,180}
\definecolor{mycolor2}{RGB}{40,75,90}
\definecolor{LUGreen}{RGB}{173,202,184}
\definecolor{LUPink}{RGB}{233,196,199}
\definecolor{LUBlue}{RGB}{0,0,128}
\definecolor{LUBronze}{RGB}{156,97,20}
\definecolor{LUblue}{RGB}{185,211,220}
\definecolor{LUGrey}{RGB}{77,76,68}

\usepackage[colorlinks = true,
linkcolor = mycolor,
urlcolor  = mycolor,
citecolor = mycolor,
anchorcolor = mycolor]{hyperref}

\newcommand\numeq[1]%
{\stackrel{\scriptscriptstyle(\mkern-1.5mu#1\mkern-1.5mu)}{=}}

\newtheorem*{definition*}{Definition}
\newtheorem*{thm*}{Theorem}
\newtheorem*{proposition*}{Proposition}
\newtheorem{proposition}{Proposition}

\def\ds{\displaystyle}

\begin{document}

\title{The Poisson-Gaussian Mixture Process: A Flexible and Robust Approach for Non-Gaussian Geostatistical Modeling}
\author{\bf{F. B. Gon\c{c}alves$^a$, M. O. Prates$^a$, G. A. S. Aguilar$^b$}}
\date{}

\maketitle

\begin{center}
{\footnotesize $^a$Universidade Federal de Minas Gerais, Brazil\\
$^b$Universidade Estadual Paulista, Brazil}
\end{center}

\begin{abstract}
\noindent This paper introduces a novel family of geostatistical models designed to capture complex features beyond the reach of traditional Gaussian processes. The proposed family, termed the Poisson-Gaussian Mixture Process (POGAMP), is hierarchically specified, combining the infinite-dimensional dynamics of Gaussian processes with any multivariate continuous distribution. This combination is stochastically defined by a latent Poisson process, allowing the POGAMP to define valid processes with finite-dimensional distributions that can approximate any continuous distribution. Unlike other non-Gaussian geostatistical models that may fail to ensure validity of the processes by assigning arbitrary finite-dimensional distributions, the POGAMP preserves essential probabilistic properties crucial for both modeling and inference.
We establish formal results regarding the existence and properties of the POGAMP, highlighting its robustness and flexibility in capturing complex spatial patterns. To support practical applications, a carefully designed MCMC algorithm is developed for Bayesian inference when the POGAMP is discretely observed over some spatial domain. 

\vspace{.3cm}


\end{abstract}

\section{Introduction}

Continuous spatial statistical modeling, also known as geostatistics, provides a probabilistic framework to analyze spatially distributed data observed over continuous domains. Spatial prediction is a fundamental problem across diverse fields such as petroleum engineering, civil engineering, mining, geography, geology, environmental science, hydrology, and climate studies \citep{dubrule1989review,hohn1998geostatistics,cressie2015statistics,bevilacqua2021non}. In geostatistical problems, one typically observes a partial realization of an underlying random field indexed by spatial locations, with the goal of understanding the stochastic dynamics of this generating process to predict values at unobserved locations or regions. The underlying random field is often modeled as a Gaussian process (GP), which provides a convenient and tractable framework for inference and prediction due to its well-understood probabilistic structure. A GP is fully characterized by its mean and covariance functions, implying that the joint distribution at any finite collection of locations is multivariate normal. This stability under summation and conditioning allows GPs to possess attractive properties, such as closure under marginalization and conditioning, and makes them a natural choice for spatial modeling.

The GP inherits useful properties from the normal distribution, making it a powerful and versatile tool for many real-world problems. However, in practice, the assumption of normality can be limiting when data exhibit skewness, heavy tails, or other non-Gaussian features. This has motivated the development of more flexible classes of geostatistical models that can accommodate non-Gaussian characteristics while retaining some desirable probabilistic properties of GPs.

Historically, non-Gaussian geostatistical processes have been developed using two main strategies. The first approach defines the process by specifying its finite-dimensional distributions (FDDs), often as a specific class of parametric distributions like the skew-normal or skew-$t$. However, this can sometimes result in invalid processes that fail to define a coherent probability measure. For example, \citet{mahmoudian2018existence} discusses necessary conditions for a skew-normal distribution to define a valid process, highlighting several inconsistencies in the literature

The second approach constructs non-Gaussian processes as functions of GPs combined with a finite set of random variables. The validity of this construction stems from the well-defined GPs and the finite-dimensionality of the other variables. For instance, \citet{mahmoudian2018existence} leveraged this approach to define processes with skew-normal FDDs. Similarly, \citet{de1997bayesian} used non-linear monotonic transformations of Gaussian fields to introduce non-Gaussian behavior. \citet{palacios2006non} developed the Gaussian-log-Gaussian model, a scale mixture of GPs capable of modeling heavy tails by using two independent GPs.

\citet{alodat2009skew} introduced a skew-Gaussian process by taking a function of a Gaussian process and a single univariate standard normal variable. Later, \citet{alodat2014extended} generalized this to an extended skew-Gaussian process by allowing the univariate normal component to have a non-zero mean. \citet{bevilacqua2021non} proposed a heavy-tailed process based on a scale mixture of a GP where the mixture component is a function of $\nu$ independent GPs, resulting in Student-$t$ marginal FDDs. The authors derived the bivariate case density but found that it includes an intractable infinite sum term, which complicates its practical implementation for likelihood-based methods.

To accommodate skewness as well as heavy tails, \citet{bevilacqua2021non} modified their model by replacing the GP with the skew-Gaussian process from \citet{zhang2010spatial}, which itself is defined using two independent GPs. The resulting process has skew-$t$ marginals, but it retains the computational challenges of the earlier model. Indeed, the authors acknowledge the impracticality of applying likelihood-based methods for even moderately sized datasets. 

Finally, \citet{tagle2020hierarchical} proposed a spatial skew-$t$ model that partitions the spatial domain and models the process as a function of a GP alongside a finite collection of random variables, each following gamma, half-normal, and multivariate normal distributions. Their model assumes conditional independence across regions, resulting in discontinuities at the boundaries of the partitioned spatial domain.

A recurring issue in non-Gaussian geostatistical models is parameter identifiability, particularly for skewness and tail parameters. \citet{Genton12} suggest a solution by formulating processes that involve an uncountable set of independent univariate random variables; however, this results in everywhere discontinuous paths that lack key regularity properties essential for statistical inference. Such processes are dense in the real line and suffer from measurability issues, restricting the use of certain desirable statistical functions.

Gaussian copulas offer another framework for defining non-Gaussian processes, allowing for flexible marginal distributions while maintaining Gaussian dependence structures. However, model elicitation and interpretability are often compromised, as copula-based models are non-linear functions of GPs. Non-Gaussian copulas can also be considered, but closed-form expressions are rarely available for moments or densities. Examples of copula-based geostatistical processes can be found in \citet{bardossy2006copula,bardossy2008geostatistical,kazianka2010copula,kazianka2011bayesian,prates2015transformed,hughes2015copcar,graler2010copulas,graler2011pair}.

While defining non-Gaussian processes as functions of GPs and other random variables ensures validity, it introduces significant challenges. Models based on multiple GP components increase computational complexity, and those relying on uncountable random variables lack regularity, impacting both model inference and interpretability. Recently, \citet{zheng2021nearest} proposed a scalable approach to non-Gaussian modeling using nearest neighbor methods with conditional distributions based on bivariate mixtures. However, FDDs and moments are only available for a few simple cases, the previously mentioned regularity issues persist and defining a process with the desired flexible properties is not straightforward.

In response to the limitations in current approaches to non-Gaussian geostatistics, this paper proposes a novel family of geostatistical models, termed the Poisson-Gaussian Mixture Process (POGAMP). The POGAMP family provides flexibility for non-Gaussian features while preserving desirable properties like path continuity and functions measurability. This model family is structured hierarchically, using an augmented Poisson process (PP) that determines the FDD at specific locations based on an initially chosen multivariate distribution, while the infinite-dimensional remainder follows a Gaussian measure. Typical multivariate distributions for POGAMP include skew-normal, student-t, skew-t, and copulas, providing an intuitive interpretation of the impact of non-Gaussian features on the overall process.

The key advantages of the POGAMP model are as follows: (1) it is always valid, regardless of the chosen multivariate distribution; (2) it is highly flexible, with FDDs that approximate the chosen distribution as the PP rate increases; (3) it maintains regularity properties such as path continuity and differentiability, and the measurability of relevant functions.

The regularity properties of the POGAMP process allows for complex applications where the geostatistical process is latent and indexes the distribution of observed variables in non-standard ways. For instance, it can be used in point pattern analysis where the process intensity is a function of the geostatistical model \citep{gonccalves2018exact} or in level-set models where spatial partitions defined by the process' levels index observed distributions \citep{gonccalves2020exact}.

An MCMC algorithm is proposed for Bayesian inference based on POGAMP observations at a finite set of locations. The algorithm uses an infinite-dimensional Markov chain that converges to the exact posterior distribution of all the unknown components of the model. To ensure exactness of the limiting distribution, the algorithm employs retrospective sampling techniques \citep{papaspiliopoulos2004retrospective}.

This paper is structured as follows: Section 2 introduces the POGAMP model, establishing its existence and main properties. Section 3 presents the MCMC algorithm for Bayesian inference under discrete POGAMP observations. Section 4 introduces a nearest neighbor GP approach to deal with large datasets. 
Section 5 brings final remarks.

\section{A novel family of geostatistical models}

\subsection{The Poisson-Gaussian mixture process}

The Poisson-Gaussian Mixture Process (POGAMP) is a hierarchically defined model leveraging conditional and marginal measures to construct a joint probability measure for geostatistical analysis. This hierarchical approach provides both tractability and intuitive interpretability of the model's components. However, ensuring the existence of the joint measure is crucial for defining a well-posed geostatistical process. We formally introduce the POGAMP in Definition \ref{defexistY} and establish its existence in Theorem \ref{theoexistence}. Before proceeding, we first define a Gaussian process.

\begin{defn}
	\label{defGP}
	A Gaussian process $Y$ is a stochastic process over a spatial domain $S \subset \mathds{R}^d$, $d \in \mathds{N}$, such that, for any finite collection of points $s_1, \dots, s_n \in S$, we have
	\begin{align}
		(Y(s_1), \dots, Y(s_n)) \sim N_n(\mu(s_1,\ldots,s_n), \Sigma(s_1,\ldots,s_n)),
	\end{align}
	where $N_n$ denotes the $n$-dimensional normal distribution, $\mu(s_1,\ldots,s_n) \in \mathds{R}^n$ is the mean vector, and $\Sigma(s_1,\ldots,s_n)$ is an $n \times n$ positive-definite covariance matrix.\\
	When $\mu(s_1,\ldots,s_n)$ is constant, and the covariance $cov(Y(s_i), Y(s_j))$ depends only on $|s_i - s_j|$ through a function $\rho$ (with $cov(Y(s_i), Y(s_j)) = \sigma^2\rho(|s_i - s_j|)$ for some $\sigma^2 > 0$), $Y$ is called a stationary and isotropic Gaussian process.
\end{defn}

To define the POGAMP, let \( S_N = (S_{N,1}, \ldots, S_{N,|N|}) \) represent a finite set of locations in \( S \subset \mathds{R}^d \), and let \( f \) denote a class of distributions that, for any such set \( S_N \), specifies a continuous multivariate distribution with density \( f_N \), support \( \mathcal{X}^{|N|} \), and an associated positive definite covariance matrix. Assume that \( f_N(y) \) is continuous in \( (S_N, y) \in (S^{|N|} \times \mathcal{X}^{|N|}) \) for any \( |N| \in \mathds{N} \), and that it is uniformly integrable with respect to the \( |N| \)-dimensional Lebesgue measure. Finally, define \( \mathcal{G}_Y \) as the probability measure of a Gaussian process on \( S \subset \mathds{R}^d \) with mean function \( \mu \) and covariance function \( \Sigma \).

\begin{defn}
	\label{defexistY}
	\textbf{The Poisson-Gaussian Mixture Process (POGAMP)}. Let \( (Y, N) \), where \( Y = \{ Y(s) : s \in S \subset \mathds{R}^d \} \) for a compact domain \( S \), be a stochastic process defined as follows:
	\begin{itemize}
		\item[i)] \( N \) is a Poisson process on \( S \) with intensity \( \lambda := \{\lambda(s) : s \in S\} \), generating events \( S_N \);
		\item[ii)] conditional on \( N \), the random variable \( Y_N := (Y(S_{N,1}), \ldots, Y(S_{N,|N|})) \) has density \( f_N \);
		\item[iii)] conditional on \( (N, Y_N) \), \( \{ Y(s) : s \in S \setminus N \} \) follows the conditional Gaussian measure induced by the measure \( \mathcal{G}_Y \).
	\end{itemize}
\end{defn}

We refer to the Gaussian process with mean function \(\mu\) and covariance function \(\Sigma\) as the base GP. The POGAMP model offers substantial flexibility through the choice of the density \(f_N\) and the intensity function \(\lambda\). The density \(f_N\) should be chosen from a class $f$ of distributions with an appropriate spatial covariance structure; for instance, a multivariate skew-\(t\) distribution with isotropic covariance. The intensity function \(\lambda\) may be constant or spatially varying.

The roles of $f$ and $\lambda$ are distinct in defining the POGAMP. The distribution $f$ imbues the finite-dimensional distributions (FDDs) of $Y$ with characteristics like skewness and heavy tails, while $\lambda$ controls how closely the FDDs of $Y$ match $f$ across space. Smaller values of $\lambda$ result in FDDs that approximate the Gaussian process characteristics, while larger values cause $Y$ to align more closely with $f$. We consider both the homogeneous case, where $\lambda(s) = \lambda \in \mathds{R}^+$ for all $s$, and a parametric non-homogeneous case $\lambda(s) = \lambda(s; \theta_{\lambda})$. Non-parametric intensity functions can also be incorporated following methodologies such as those in \citet{gonccalves2018exact} and \citet{gonccalves2020exact}.

The uniform integrability condition on $f_N$ is required to ensure the weak convergence of the FDDs of the POGAMP to $f$, as stated in Theorem \ref{theo2}. This condition is typically met, as with the skew-t distribution, and can be verified if $f^{1+\delta}$ is integrable for some $\delta > 0$.

The existence of the POGAMP is established in the following theorem.

\begin{thm}
	\label{theoexistence}
	\textbf{Existence of the POGAMP}. Consider the measurable space $(\mathcal{Y} \times \mathcal{N}, \sigma(\mathcal{B}(\mathcal{Y}) \times \mathcal{B}(\mathcal{N})))$, where $\mathcal{Y}$ is the Banach space of continuous real-valued functions on $S$, $\mathcal{B}(\mathcal{Y})$ is the Borel $\sigma$-algebra generated by the open sets of $\mathcal{Y}$ in the strong topology, $\mathcal{N}$ is the space of all finite subsets of $S$ and $\mathcal{B}(\mathcal{N})$ is the smallest $\sigma$-algebra that makes all counting functionals measurable. Then, the stochastic process $(Y,N)$ in Definition \ref{defexistY} defines a probability measure $\mathcal{P}$ on $(\mathcal{Y} \times \mathcal{N}, \sigma(\mathcal{B}(\mathcal{Y}) \times \mathcal{B}(\mathcal{N})))$, implying that $Y$ (under $\mathcal{P}$) is a valid stochastic process.
\end{thm}

The proof of Theorem \ref{theoexistence} relies on the following lemma.

\begin{lemma}
	\label{lemma1}
	Let $(\Omega_1, \mathcal{F}_1, \mu_1)$ be a probability space such that, for each $\omega_1 \in \Omega_1$, there exists a probability measure $\mu_{2,\omega_1}$ on $(\Omega_2, \mathcal{F}_2)$. Consider the joint measurable space $(\Omega, \sigma(\mathcal{F}))$, where $\Omega = \Omega_1 \times \Omega_2$ and $\mathcal{F} = \mathcal{F}_1 \times \mathcal{F}_2$, and suppose that
	\begin{align}
		g_{A}(\omega_1) = \int_{\Omega_2} \mathbbm{1}[(\omega_1,\omega_2)\in A]d\mu_{2,\omega_1}(\omega_2),
	\end{align}
	is $\mathcal{F}_1$-measurable for all $A \in \mathcal{F}$. Then, there exists a probability measure $\mu$ on $(\Omega, \sigma(\mathcal{F}))$ satisfying
	\begin{align}
		\mu(A) = \int_{\Omega_1} \int_{\Omega_2} \mathbbm{1}[(\omega_1, \omega_2) \in A] d\mu_{2,\omega_1}(\omega_2) d\mu_1(\omega_1),
	\end{align}
	for all $A \in \mathcal{F}$, where $\mu$ is called the joint measure.
\end{lemma}

Lemma 1 establishes the necessary conditions for a pair consisting of a marginal and a conditional probability measure to define a valid joint probability measure. Unlike similar results in the literature, this lemma applies to a more general setting rather than being restricted to finite-dimensional real measurable spaces. The proofs for Lemma 1 and Theorem 1 are provided in Appendix A.

\subsection{Properties of the POGAMP}\label{POG_prop}

\subsubsection{Absolute continuity with respect to GP measure and measurability}\label{ACM}

Let $\mathcal{G}$ be the probability measure of the POGAMP under which $f_N$ is the normal distribution induced by the base GP. We refer to the process $(Y, N)$ under $\mathcal{G}$ as the augmented GP, as the induced marginal measure of $Y$ is the base GP.

\begin{proposition}
	\label{prop1}
	The measure $\mathcal{P}$ is absolutely continuous w.r.t. $\mathcal{G}$, with the Radon-Nikodym derivative $\cfrac{d\mathcal{P}}{d\mathcal{G}}=\cfrac{f_N}{g_N}(Y_N)$, where $g_N$ is the Lebesgue density of $Y_N$ under $\mathcal{G}$.
\end{proposition}

The proof of Proposition \ref{prop1} is provided in Appendix \ref{appendixA}. An equivalent statement is that $\mathcal{G}$-a.s. implies $\mathcal{P}$-a.s. Consequently, properties of GPs, such as almost sure continuity and differentiability, are inherited by the POGAMP. This result also facilitates the derivation of an MCMC algorithm for inference on discretely observed POGAMPs.

The measurable space in which the POGAMP is defined supports a broad set of measurable functions, including specific functions of $Y$ on the entire space $S$, e.g., $\int_S g(Y(s)) ds$ for any Borel function $g$, $\sup_{s \in S} Y(s)$, and $\inf_{s \in S} Y(s)$. Measurability of these functions is crucial for tasks such as statistical estimation and defining complex spatial processes, where the POGAMP is be fully latent, such as when $g(Y(s))$ represents the intensity function of a Poisson process over $S$.

Departures from normality, often due to skewness or heavy tails, can be managed using general multivariate distributions like skew-normal, student-t, and skew-t. However, extending such distributions to infinite-dimensional processes poses challenges. \citet{Genton12} addresses identifiability issues here, relying on processes defined as functions of a GP and a set of i.i.d. random variables $\{Z(s); s \in S\}$. Kolmogorov's Existence Theorem guarantees the existence of such processes only in a poor measurable product space. Consequently, only functions determined by a countable collection of locations are measurable, excluding, for instance, integrals of non-trivial Borel functions like $\int_S g(Y(s)) ds$ and sets like $A(Y) = \{s \in S; Y(s) \in B \subset \mathds{R}\}$. Furthermore, these processes exhibit highly irregular paths and are almost surely everywhere dense on $\mathds{R}$. Appendix \ref{appendixD} presents the definitions of the processes proposed in \citet{Genton12} and shows realizations of those processes to illustrate the path roughness. 

The POGAMP provide an efficient alternative, enabling identifiable, non-Gaussian geostatistical processes with path continuity and nice measurability properties. Additionally, POGAMP versions of the processes described in \citet{Genton12} depend only on a finite-dimensional set of i.i.d. random variables (and a GP), favoring MCMC efficiency for inference.

Another result following from Proposition \ref{prop1} is that the Kullback-Leibler divergence between $\mathcal{P}$ and $\mathcal{G}$ depends on $Y$ only through $Y_N$.
\begin{proposition}
	\label{prop2}
	The Kullback-Leibler divergence between $\mathcal{P}$ and $\mathcal{G}$ is
	\begin{align}
		D_{\text{KL}}(\mathcal{P} \parallel \mathcal{G}) = \sum_{n=0}^{\infty} \frac{\rm{e}^{-\Lambda_S} (\Lambda_{S})^n}{n!} \int_{S^{n}} \int_{\mathcal{X}^{n}} \prod_{j=1}^{n} \frac{\lambda(s_{j})}{\Lambda_S} \log \left(\frac{f_{N}}{g_{N}}(y)\right) f_{N}(y) \, \text{d}y \, \text{d}s,
	\end{align}
where $\Lambda_S = \int_S \lambda(s) \, ds$, $s = (s_{1}, \ldots, s_{|N|})$, $y = (y_1, \ldots, y_{|N|})$, and $f_{N}$ and $g_{N}$ are the conditional densities of $Y$ at $s$, under $\mathcal{P}$ and $\mathcal{G}$, respectively, given $N$. Let $\mu(S)$ be the volume of $S$. If $N$ is a homogeneous PP with rate $\lambda$, we have
	\begin{align}
		D_{\text{KL}}(\mathcal{P} \parallel \mathcal{G}) = \sum_{n=0}^{\infty} \frac{\rm{e}^{-\lambda \mu(S)} \lambda^n}{n!} \int_{S^{n}} \int_{\mathcal{X}^{n}} \log \left(\frac{f_{N}}{g_{N}}(y)\right) f_{N}(y) \, \text{d}y \, \text{d}s.
	\end{align}
\end{proposition}
The proof of Proposition \ref{prop2} is presented in Appendix \ref{appendixA}.

\subsubsection{Convergence of the finite-dimensional distributions}

One significant advantage of the POGAMP is that its finite-dimensional distributions inherit properties of the distribution $f$. As the IF of $N$ increases, the FDDs converge to $f$, as demonstrated in the following theorem.

\begin{thm}\label{theo2}
	\textbf{Convergence of the POGAMP}. Let $\{\lambda_n\}_{n=1}^{\infty}$ be a sequence of positive real numbers such that $\lambda_n \uparrow \infty$ and $(\lambda_{n+1} - \lambda_n) > \epsilon > 0$ for all $n \in \mathds{N}$. For any $r \in \mathbbm{N}$, let $S_R$ be a set of $r$ distinct locations in $S$, and define an infinite sequence of POGAMPs with common $f$ distribution, where $N$ in the $n$-th POGAMP is a Poisson process with $\lambda_n(s) \geq \lambda_n$ for all $s \in S$. Define $Y_{n,R}$ as the component $Y$ of the $n$-th POGAMP at $S_R$ and $Y_{f,R}$ as a $r$-dimensional random variable with distribution $f$ at $S_R$. Then, $Y_{n,R} \xrightarrow{d} Y_{f,R}$.
\end{thm}

The proof of Theorem \ref{theo2} is provided in Appendix \ref{appendixA}. 

For cases where $f$ defines a valid process, the POGAMP can approximate this process (in terms of FDDs) by increasing $\lambda$, while retaining the GP measure's tractability. Even if $f$ does not define a valid process, the POGAMP's FDDs can still get arbitrarily close to $f$. Figure \ref{fig1} illustrates this result, showing the empirical marginal density of the POGAMP at the centroid of a square of size 10 for various $\lambda$ values, with $f$ distribution being skew-normal, student-t, and skew-t.

\begin{figure}
\centering
\includegraphics[width=0.32\linewidth]{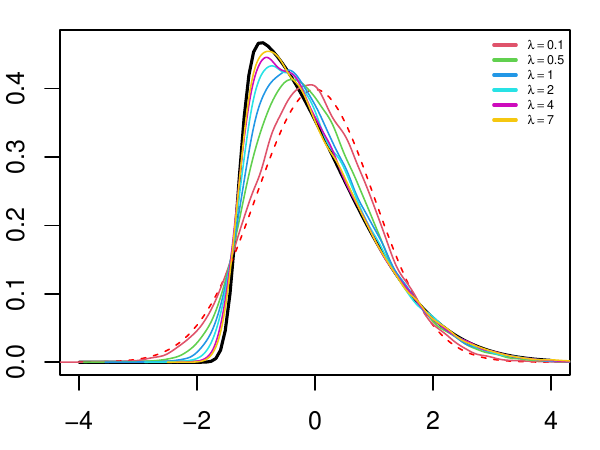}
\includegraphics[width=0.32\linewidth]{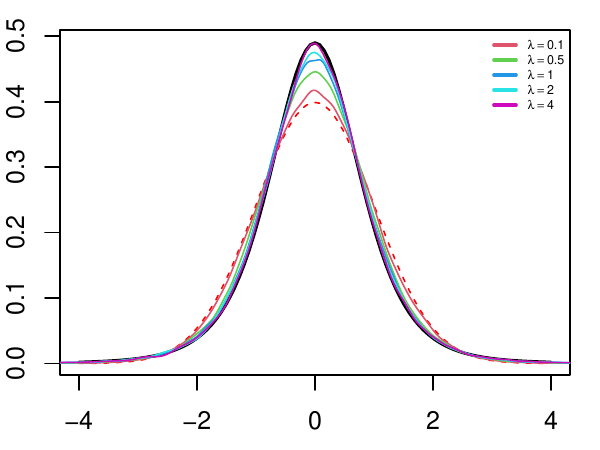}
\includegraphics[width=0.32\linewidth]{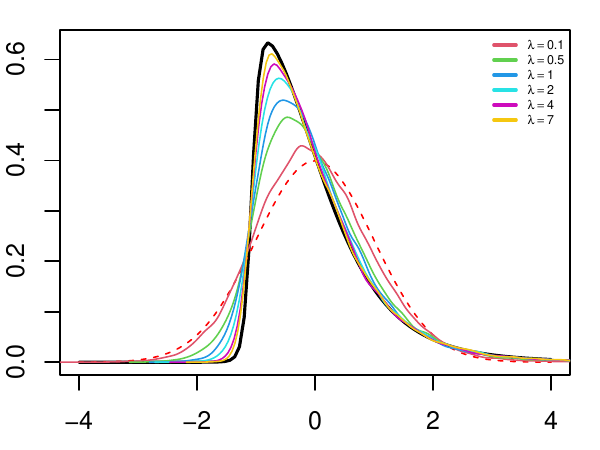}
\caption{Empirical density of the POGAMP at the centroid of a square with side 10 for skew-normal, student-t, and skew-t $f$ distributions, for various $\lambda$ values. The black line shows the marginal $f$ density, and the dashed red line shows the base GP's marginal density.}
\label{fig1}
\end{figure}

\subsubsection{Finite-dimensional distributions and covariance function}

Standard calculations based on the hierarchical structure of the POGAMP are used to obtain the density of its finite-dimensional distributions. We use $\mathlarger{\pi}$ as a general notation for densities, where the implicit dominating measure is the counting measure for discrete variables and the Lebesgue measure for continuous ones.

\begin{proposition}
	\label{prop3}
	\textbf{Density of the finite-dimensional distributions}.
	For any finite collection of locations $S_R$, the density $\mathlarger{\pi}_R$ of $Y_R$ (i.e., $Y$ at $S_R$), under the POGAMP measure, is given by
	\begin{align}
		\label{explicitden}
		\mathlarger{\pi}_R(y) &= \sum^{\infty}_{n=0} w_n g(y;n),
	\end{align}
	where
	\begin{align}
		w_n &= \frac{\left[\Lambda_S\right]^{n}}{n!} c_{n}, \quad g(y;n) = \frac{g^*(y;n)}{c_n}, \\
		c_n &= \int_{\mathbbm{R}^n} g^*(y;n) \, \text{d}y, \quad g^*(y;n) = \int_{S^{n}} \int_{\mathbbm{R}^{n}} \mathlarger{\pi}_{\mathcal{G}}(y \mid x, s, n) \, f_N(x \mid s, n) \, \mathlarger{\pi}(s\mid n) \, \text{d}x \, \text{d}s,
	\end{align}
	with $s$ and $x$ being $n$-dimensional vectors, $\mathlarger{\pi}(s\mid n) = \prod_{j=1}^{n} \frac{\lambda(s_{j})}{\Lambda_S}$, $f_N(x \mid s, n)$ being the density of the $f$ distribution at locations $s$, and $\mathlarger{\pi}_{\mathcal{G}}(y \mid x, s, n)$ being the density of $Y$ at locations $S_R$, conditional on $(Y_N, S_N, |N|) = (y, s, n)$, under the base GP measure.
\end{proposition}

Note that the marginal density of $Y_R$ is a discrete mixture of the conditional (on $|N|$) densities $g$ which, in turn, are continuous location-scale-mixture of normals where the linear components in the conditional mean follow a $f_N$ distribution and the set of conditional locations follow a Poisson process.

Next, we express the covariance function of $Y$ under the POGAMP measure.

\begin{proposition}
	\label{prop4}
	Let $s_1$ and $s_2$ be any two locations in $S$, $\Sigma_{i}$ be the row vector of covariances between $Y_i:=Y(s_i)$ and $Y(S_{N})$ under the base GP, for $i=1,2$, and $\Sigma_{N}$ and $\Sigma_{N,f}$ be the covariance matrices of $Y_N$ under the base GP and the $f$ distribution, respectively. Then,
	\begin{equation}
		\operatorname{Cov}(Y_1, Y_2) = \mathrm{E}_N[\operatorname{Cov}(Y_1, Y_2 \mid Y_N)] + \mathrm{E}_N[\Sigma_{1} \Sigma^{-1}_{N} \Sigma_{N,f} (\Sigma_{2} \Sigma^{-1}_{N})^\intercal].
	\end{equation}
	If the covariance function is the same under the base GP and $f$, then it is the covariance function of the POGAMP.
\end{proposition}

The proof of Proposition \ref{prop4} is presented in Appendix \ref{appendixA}.

The second result in Proposition \ref{prop4} has important practical implications as it provides an analytical representation of the covariance function of the POGAMP, allowing for a clear interpretation of this.

\subsubsection{Spatial Symmetry}

We focus here on the case where $S \subset \mathds{R}^2$, as is typical in geostatistical applications, and state an interesting property regarding the spatial symmetry of the POGAMP. First, consider the following definitions.

\begin{defn}
	\label{defdsymm1}
	\textbf{Symmetry of a Compact Region}. We say that a compact region $S \subset \mathds{R}^2$ is symmetric if there exists at least one rotation of $\mathds{R}^2$ (that is not a multiple of $2\pi$) that leaves $S$ unchanged.
\end{defn}

\begin{defn}
	\label{defdsymm2}
	\textbf{Symmetry of a Poisson Process}. We say that a Poisson process $N$ on a symmetric region $S \subset \mathds{R}^2$ is symmetric if all rotations of $\mathds{R}^2$ that preserve $S$ also preserve the intensity function of $N$.
\end{defn}

\begin{defn}
	\label{defdsymm3}
	\textbf{Symmetry of Locations}. Suppose that $S$ and $N$ are symmetric. Let $\mathbf{s}$ and $\mathbf{s}'$ be two sets of $r$ locations in $S$, for any $r \in \mathds{N}$. We say that $\mathbf{s}$ and $\mathbf{s}'$ are symmetric with respect to $(S, N)$ if, for at least one of the rotations defined in Definition \ref{defdsymm1}, $\mathbf{s}$ in the rotated space is equal to $\mathbf{s}'$ in the non-rotated one.
\end{defn}

\begin{proposition}
	\label{prop5}
	Suppose that $f$ is stationary and that $S$ and $N$ are symmetric. Then, for any two sets $\mathbf{s}$ and $\mathbf{s}'$ in $S$ that are symmetric with respect to $(S, N)$, we have that $Y(\mathbf{s}) \overset{d}{=} Y(\mathbf{s}')$ under the POGAMP measure.
\end{proposition}

The proof of Proposition \ref{prop5} is presented in Appendix \ref{appendixA}.

We present three examples of the symmetry property of the POGAMP in Figure \ref{sym_fig}. In Figure \ref{sym_fig}(a), $N$ is a homogeneous PP, and the four sets of points defined by different colors are symmetric with each other. In Figure \ref{sym_fig}(b), the intensity function of $N$ is piecewise constant; the blue and red sets are symmetric, as are the black and green ones. In Figure \ref{sym_fig}(c), the intensity function of $N$ is proportional to a bivariate normal density function centered at the centroid, where the four sets of points defined by different colors are symmetric with each other.

\begin{figure}[!ht]
\centering
\includegraphics[width=0.85\linewidth]{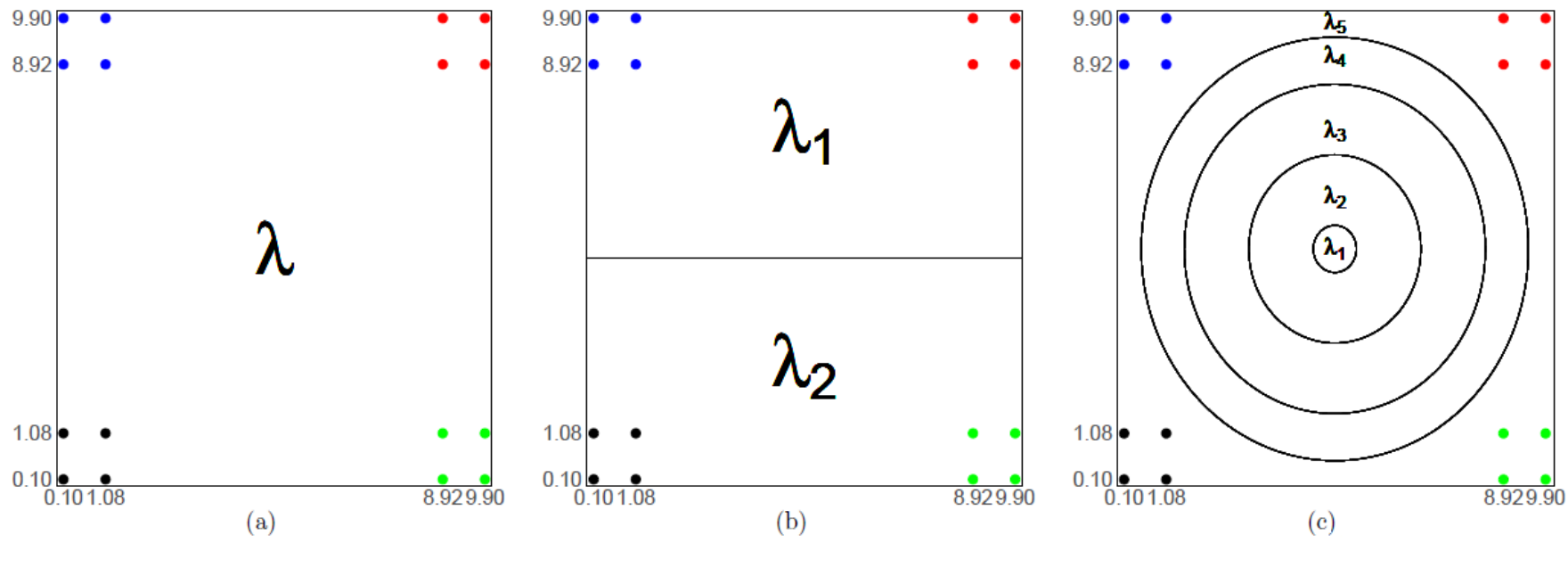}
	\caption{Examples of spatial symmetry in POGAMPs: (a) Homogeneous PP with symmetric point sets; (b) Piecewise constant intensity function with symmetric color-coded point sets; (c) Intensity function proportional to a bivariate normal density with symmetry at the centroid, with color-coded symmetric point sets.}\label{sym_fig}
\end{figure}

\subsection{Model specification and identifiability}\label{mod_spec}

The POGAMP contributes to modeling flexibility in two main ways. First, it allows the definition of a process with finite-dimensional distributions that can closely approximate any class of multivariate distributions. Second, as discussed in Section \ref{ACM}, the rich probability space underlying the POGAMP yields desirable statistical properties, making it viable as a fully latent process in complex models. Additionally, the POGAMP permits continuous variation of local behavior across $S$ between Gaussian and $f$-like distributions, achievable through a non-constant IF $\lambda$. Specifically, $\lambda$ can be parametrically estimated, using data to inform the local variation of behavior.

The flexibility in choosing both the $f$ distribution and intensity function $\lambda$ establishes the POGAMP as a versatile class of geostatistical models. Consequently, model specification and identifiability become crucial considerations in statistical analyses with the POGAMP. Model parsimony is also essential, and a balance between flexibility and predictive power should be sought.

The main step in specifying a POGAMP involves selecting $f$, the distribution that incorporates desired non-Gaussian characteristics, such as skewness or heavy tails. The finite-dimensional distributions of the POGAMP mix the properties of $f$ with those of the base Gaussian process, so understanding the moments and densities of $f$ is critical to anticipate the process dynamics.

One key consideration is whether to match moments (marginal and joint) of the $f$ distribution and the base GP, especially mean, variance, and covariance functions. This impacts the model's complexity and parsimony. Moment matching is generally a sound strategy, particularly for stationary models. For non-stationary models, one could achieve spatial variation either by varying $\lambda$ spatially or by mismatching moments between $f$ and the GP, perhaps using a non-stationary $f$.

Typical choices for $f$ include distributions known for skewness and heavy tails, like skew-Normal, student-t, skew-t, and copula-based distributions. Following \citet{Genton12}, we consider families that ensure parameter identifiability for skewness and degrees of freedom. These families have been used in areal settings by \citet{prates2012dengue} and are presented in Appendix \ref{appendixB}. Like traditional Gaussian process models with a Matérn covariance function, variance and range parameters are not individually identifiable. Although it does not hinder prediction \citep[see][]{Zhang04}, the lack of identifiability can slow down the MCMC. We mitigate this issue by adopting the joint penalizing prior for the pairs of variance and range parameters in the base GP and in the $f$ distribution proposed in \citet{fuglstad2019constructing}.

Finite mixture distributions are also a promising choice for $f$, particularly mixtures with fixed skewness and tail properties. Such mixtures allow heavy tails and skewness to be estimated through mixture weights rather than individual parameters \citep[see][]{bispo2020}. Spatial covariates can also be included by defining $Y(s) = g(\mathbf{X}(s);\theta) + \xi(s)$, where $g$ is a function of spatial covariates $\mathbf{X}$ and parameters $\theta$, and $\xi$ is a mean-zero POGAMP. Observation errors can be considering by making $\ddot{Y}(s)=Y(s)+e(s)$, where $\ddot{Y}(s)$ are the observed values and the $e(s)$ are i.i.d., possibly normal, errors, and $Y$ is a POGAMP.

Features such as non-stationarity and anisotropy can be incorporated through $f$, thereby eliminating the need to embed them directly in the base GP and reducing the associated computational complexity. When the dimension of $N$ significantly increases computational costs, conditional independence structures can be employed to mitigate the issue without sacrificing modeling flexibility.

Identifiability poses challenges when balancing non-Gaussian features in $f$ with the degree to which POGAMP FDDs resemble $f$ (determined by $\lambda$). For instance, heavy tails or skewness in the POGAMP can result from adjusting either $f$ or $\lambda$. Specifically, when $f$ is a student-t distribution, heavy tails may be achieved by changing its degrees of freedom or adjusting $\lambda$ for a fixed heavy-tailed $f$ (Figure \ref{t5_3_fig}). Though different parameter settings may not yield identical distributions, they can produce very similar results, leading to potential identifiability issues, especially with limited data.

\begin{figure}[!ht]
\centering
\includegraphics[width=0.6\linewidth]{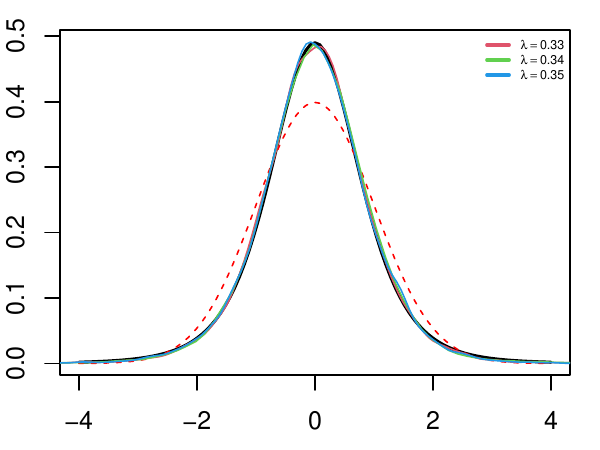}
\caption{Empirical marginal density of the POGAMP with a student-t $f$ with 3 degrees of freedom for different $\lambda$ values ($E[|N|]=100\lambda$). The black line shows a student-t density with 5 degrees of freedom.}\label{t5_3_fig}
\end{figure}

In a Bayesian framework, identifiability can be enhanced using robust penalizing priors. When balancing identifiability and computational efficiency, scenarios with smaller values of $\lambda$ are generally preferred. Consequently, we adopt a penalizing prior distribution for $\lambda$, in the homogeneous case, that aligns with the expected scale of $|N|$. As indicated by Figures \ref{fig1} and \ref{t5_3_fig}, having $|N| < 300$ should typically suffice for a good fit, provided suitable estimates of the non-Gaussian feature parameters are obtained. In this context, a $\text{Gamma}(10,10)$ prior on $\lambda$ is a reasonable choice to penalize larger values of $|N|$.

Another possibility is to adopt a penalizing prior on the complexity parameters of $f$, such as degrees of freedom and skewness.
For a given complexity parameter $\nu$ indexing $f$, we define a base value $\nu_0$ that represents the most complex allowed form of $f$ and penalize departures from $\nu_0$ towards simpler models using the PC prior from \citet{penprior2017}. In particular,
\begin{equation}\label{PCP}
\pi(\nu) = \eta e^{-\eta d(\nu,\nu_0)} \left| \frac{\partial d(\nu,\nu_0)}{\partial \nu} \right|,
\end{equation}
where $d(\nu,\nu_0) = \sqrt{2 \, KLD(\nu,\nu_0)}$, and $KLD(\nu,\nu_0)$ is the Kullback-Leibler divergence between the marginal of $f$ under $\nu$ and $\nu_0$:
\begin{equation}\label{KLD}
KLD(\nu,\nu_0) = \int_{\mathcal{X}} f(x\mid \nu) \log \left( \frac{f(x\mid \nu)}{f(x\mid \nu_0)} \right) dx,
\end{equation}
which is typically computed using numerical methods. For the degrees of freedom parameter, which controls tail heaviness, careful consideration is required to set the hyperparameter $\eta$. Lower values of this parameter increase estimation complexity, especially with limited data. 

An alternative approach involves fixing the non-Gaussian feature parameters at the "worst-case" scenario--meaning the case farthest from Gaussian behavior--and allowing $\lambda$ to control the degree of non-Gaussianity.

\section{MCMC}\label{mcmc_sec}

Performing statistical inference for infinite-dimensional models--when the model's unknown quantities (parameters and latent variables) are infinitely uncountable--is a highly complex problem. Historically, solutions relied on discrete (finite-dimensional) approximations of these quantities, introducing a significant error that is typically hard to measure or control. Advances in computational methods, especially Monte Carlo within a Bayesian approach, have brought new perspectives by allowing analyses free from discretization error, involving only Monte Carlo error. The latter is far simpler to quantify and control, yielding more precise and computationally efficient analyses.

Exact (discretization-error-free) inference solutions for infinite-dimensional problems have become possible largely due to a neat simulation technique called retrospective sampling. This technique effectively manages infinite-dimensional random variables by revealing only a finite-dimensional representation that:
\begin{enumerate}
    \item Is sufficient for executing algorithmic steps, e.g., MCMC.
    \item Allows any finite-dimensional subset of the infinite-dimensional remainder to be conditionally simulated.
\end{enumerate}
The idea of retrospective sampling in simulating infinite-dimensional random variables was introduced in \cite{beskos2005exact} for simulating diffusion paths exactly and has since been applied in numerous statistical contexts (see, for example, \cite{beskos2006exact}, \cite{gonccalves2018exact}, and \cite{gonccalves2023exact}).

We propose an infinite-dimensional MCMC algorithm using retrospective sampling to perform exact Bayesian inference for discretely-observed POGAMPs. This algorithm is a Gibbs sampler with Metropolis-Hastings steps and is exact, with the posterior distribution of all unknown quantities in the model as its invariant distribution. Each algorithmic step is carefully designed for computational efficiency and practical feasibility in tackling the inference problem. 

Suppose we observe a POGAMP process \( Y \) at a finite set of locations \( s_1, \ldots, s_n \) within the region \( S \), yielding data \( Y_o = (Y_{s_1}, \ldots, Y_{s_n}) \). The vector of unknown quantities to be estimated is \( \psi = (N, Y_N, Y_u, \theta_{\mathcal{G}}, \theta_f, \lambda) \), where \( Y_u = Y \setminus (Y_N, Y_o) \). Vectors \( \theta_{\mathcal{G}} \) and \( \theta_f \) are the sets of parameters indexing the base GP and \( f \)-distribution, respectively.

Under the Bayesian paradigm, inference about \( \psi \) is based on the posterior distribution of \( (\psi \mid Y_o ) \), which has a density, with respect to a suitable dominating measure, proportional to
\begin{align}
	\mathlarger{\pi}(Y_o, \psi) &= \mathlarger{\pi}_{\mathcal{G}}(Y_u\mid Y_N, Y_o, N, \theta_{\mathcal{G}}) \mathlarger{\pi}_{\mathcal{G}}(Y_o\mid Y_N, N, \theta_{\mathcal{G}}) \mathlarger{\pi}(Y_N\mid \theta_f, N) \mathlarger{\pi}(N\mid \lambda) \mathlarger{\pi}(\lambda) \mathlarger{\pi}(\theta_f, \theta_{\mathcal{G}}),
	\label{completejoint}
\end{align}
where \( \mathlarger{\pi}_{\mathcal{G}} \) refers to densities relative to the base GP measure, and \( \mathlarger{\pi}(Y_N \mid \theta_f, N) \) is \( f_N \).

To devise an efficient algorithm, we extend \( \psi \) by including a \( d \)-dimensional random field \( Z \) over \( S \). For the families of distributions in this paper, \( d \) takes values in \( \{0,1,2\} \), and \( Z \) is a field of independent random variables with identical distributions in each dimension. This extension facilitates a tractable definition of the \( f \)-distributions in \citet{genton2012identifiability}, allowing for a computationally feasible MCMC without complex transdimensional steps.

We define \( \theta_f = (\theta_{f,1}, \theta_{f,2}, \theta_{f,3}) \), where:

- \( \theta_{f,1} \) represents location and scale parameters. For \( d > 0 \), these parameters index the distribution of \( (Y_N \mid Z_N) \). For \( d = 0 \), they are location and scale parameters that can be decoupled from the target marginal in the copula model.

- \( \theta_{f,2} \) represents correlation parameters. For \( d > 0 \), these parameters index the distribution of \( (Y_N \mid Z_N) \). For \( d = 0 \), they index the Gaussian copula.

- \( \theta_{f,3} \) represents additional parameters. For \( d > 0 \), these are skewness and heavy-tail parameters that index the distribution of \( Z \) and/or \( (Y_N \mid Z_N) \). For \( d = 0 \), these are non-location/scale parameters indexing marginal distribution in the copula model.

Typically, prior independence is assumed among \( \theta_{\mathcal{G}} \), \( \theta_{f,1} \), \( \theta_{f,2} \), and \( \theta_{f,3} \).

We redefine the extended \( \psi \) as \( (N, Y_N, Z_N, Y_u, Z_u, \theta_{\mathcal{G}}, \theta_{f}, \lambda) \) and consider the following factorization:
\begin{align}
	\mathlarger{\pi}(Y_o, \psi) &= \mathlarger{\pi}_{\mathcal{G}}(Y_u\mid Y_N, Y_o, N, \theta_{\mathcal{G}}) \mathlarger{\pi}_{\mathcal{G}}(Y_o\mid Y_N, N, \theta_{\mathcal{G}}) \mathlarger{\pi}(Y_N\mid Z_N,N,\theta_{f}) \mathlarger{\pi}(Z_N\mid N,\theta_{f,3}) \nonumber \\
	& \mathlarger{\pi}(Z_u\mid N,\theta_{f,3}) \mathlarger{\pi}(N\mid \lambda) \mathlarger{\pi}(\lambda, \theta_f, \theta_{\mathcal{G}}),
	\label{completejoint2}
\end{align}

The Gibbs sampling algorithm considers the following blocking scheme:
\[
(N,Z),\;\;\; (Y_N,Y_u),\;\;\; Z,\;\;\; \lambda,\;\;\; (\theta_{\mathcal{G}},Y_u),\;\;\; (\theta_{f,1},\theta_{f,2}),\;\;\; \theta_{f,3},\;\;\; Y_u,\;\;\; Z_u.
\]

Retrospective sampling is used to simulate the Markov chain by unveiling \( (Y_u, Z_u) \) only at a finite (though random) subset of locations per iteration. Additionally, note that the full conditional distribution of \( Y_u \) is simply the base GP conditional on \( (Y_o, Y_N, \theta_{\mathcal{G}}) \), and the full conditional of \( Z_u \) is its prior.

Some blocks with two components focus on updating one component while including the second to ensure the algorithm's efficiency and validity. In particular, the focus on block \( (N,Z) \) is \( N \), on \( (Y_N,Y_u) \) is \( Y_N \), and on \( (\theta_{\mathcal{G}},Y_u) \) is \( \theta_{\mathcal{G}} \).

Algorithms to sample each block in the Gibbs sampler are presented below. Details for each of the four classes of \( f \)-distribution considered--Student-t, skew-normal, skew-t, and Gaussian copula--are provided in Appendix \ref{appendixB}.

\subsection{Sampling $(Y_N,Y_u)$ and $Z_N$}\label{FCYN}

For the families of $f$ distributions considered in this paper, $d$ is zero when $f$ is defined using a Gaussian copula. This implies the existence of a tractable 1-1 function $g$ such that $X_N=g(Y_N)$ follows a multivariate normal distribution under $f$. This parameterization is employed to develop an efficient algorithm to update $Y_N$ in this case. For the case $d=2$, we define $Z=(Z_1,Z_2)$ and $Z_N=(Z_{N,1},Z_{N,2})$.

The full conditional densities are as follows:
\begin{eqnarray}
	\mathlarger{\pi}(g(Y_N),Y_u\mid \cdot) &\propto& \mathlarger{\pi}_{\mathcal{G}}(Y_u,Y_o\mid g^{-1}(X_N), N, \theta_{\mathcal{G}})\mathlarger{\pi}(X_N\mid N,\theta_{f}),\;\mbox{if } d=0, \label{YNdistr} \\
	\mathlarger{\pi}(Y_N,Y_u\mid \cdot) &\propto& \mathlarger{\pi}_{\mathcal{G}}(Y_u,Y_o\mid Y_N, N,\theta_{\mathcal{G}})\mathlarger{\pi}(Y_N\mid Z_N,N,\theta_{f}),\;\mbox{if } d>0, \\
	\mathlarger{\pi}(Z_N\mid \cdot) &\propto& \mathlarger{\pi}(Y_N\mid Z_N,N,\theta_{f})\mathlarger{\pi}(Z_N\mid N,\theta_{f,3}),\;\mbox{if } d=1, \label{ZNdistr} \\
	\mathlarger{\pi}(Z_{N,1}\mid \cdot) &\propto& \mathlarger{\pi}(Y_N\mid Z_N,N,\theta_{f})\mathlarger{\pi}(Z_{N,1}\mid N,\theta_{f,3}),\;\mbox{if } d=2, \\
	\mathlarger{\pi}(Z_{N,2}\mid \cdot) &\propto& \mathlarger{\pi}(Y_N\mid Z_N,N,\theta_{f})\mathlarger{\pi}(Z_{N,2}\mid N,\theta_{f,3}),\;\mbox{if } d=2.
\end{eqnarray}

For the $f$ distributions considered, when $d>0$, $\mathlarger{\pi}(Y_N,Y_u\mid \cdot)=\mathlarger{\pi}_{\mathcal{G}}(Y_u\mid Y_o,Y_N,\theta_{G})\mathlarger{\pi}(Y_N\mid \cdot)$, where $\mathlarger{\pi}(Y_N\mid \cdot)$ is a known multivariate normal.

For $d=0$, the full conditional distribution of $(g(Y_N),Y_u)$ factorizes as\\ $\mathlarger{\pi}_{\mathcal{G}}(Y_u\mid g(Y_N),Y_o,N,\theta_{\mathcal{G}})\mathlarger{\pi}(g(Y_N)\mid Y_o,N,\theta_{\mathcal{G}},\theta_f)$. Sampling is performed via Metropolis Hastings (MH) with proposal distribution $q(g(Y_N),Y_u\mid\cdot)=\mathlarger{\pi}_{\mathcal{G}}(Y_u\mid g(Y_N),Y_o,N,\theta_{\mathcal{G}})q(Y_N\mid\cdot)$, where the second term is the pCN proposal discussed below.

The full conditional distribution of $Z_N$ is a truncated multivariate normal for the component that introduces skewness in the skew-normal and skew-t distributions. In this case, sampling is done using an embedded Gibbs sampling algorithm. Details about this distribution and sampling method are provided in Appendix \ref{appendixB}.
For cases where $Z_N$ introduces heavy tails in the student-t and skew-t distributions, a non-reversible MH algorithm is applied.

While a (centered) Gaussian random walk would typically be the default proposal for the MH steps of $g(Y_N)$ and $Z_N$, this approach is known to underperform as the dimension of the Markov chain grows \citep{roberts1997weak,cotter,kamatani}. This issue is exacerbated in our context, as $|N|$ will typically be large. Alternative proposals that achieve better performance under high-dimensional target distributions have been studied extensively. In particular, the pCN proposal \citep{cotter} is advantageous when the prior distribution is Gaussian, as is the case for $g(Y_N)$ in the Gaussian copula model. Unlike the centered Gaussian random walk, this proposal remains valid in infinite dimension (for a Gaussian process prior) and is reversible with respect to the prior, thus cancelling out the prior density in the MH acceptance probability expression. For cases where the prior is non-normal, the non-reversible MH algorithm proposed by \citet{kamatani} is more suitable. This algorithm is used for components of $Z_N$ that introduce heavy tails and have a Gamma prior distribution, which induces a full conditional distribution with positive excess kurtosis. Specific strategies, including a reparametrization of $Z_N$, are needed to leverage the full potential of this proposal, as described below.

Let $\Lambda$ be the covariance matrix of $\mathlarger{\pi}(g(Y_N)\mid N,\theta_f)$, and define $\zeta_N\sim \text{Normal}(0,\Lambda)$. The pCN proposal for $g(Y_N)$ is:
\begin{equation}
	g(\ddot{Y}_{N}) = \sqrt{(1-\beta^2)}g(Y_N) + \beta\zeta_N.
\end{equation}
The tuning parameter $\beta\in(0,1)$ is chosen to achieve an acceptance rate of approximately 0.234 \citep{cotter}. The acceptance probability for a move from $(g(Y_N),Y_u)$ to $(g(\ddot{Y}_N),Y_u)$ is
\begin{equation}
 1 \land \frac{\mathlarger{\pi}_{\mathcal{G}}(Y_o\mid g(\ddot{Y}_{N}),N,\theta_{f},\theta_{\mathcal{G}})}{\mathlarger{\pi}_{\mathcal{G}}(Y_o\mid g(Y_N),N,\theta_{f},\theta_{\mathcal{G}})}.
\end{equation}

The $\Delta$-guided Metropolis-Haar kernel of \citet{kamatani} extends the state space with a binary random variable $V$ taking values in $\{-1,1\}$. We consider a reparameterization $W_{N}=h(Z_N)$ so that the full conditional distribution of $W_M$ is more symmetric than that of $Z_N$, which improves the efficiency of the $\Delta$-guided mixed pCN kernel we adopt.

Define $\Delta u=(u-\mu_W)'\Sigma_{W}^{-1}(u-\mu_W)$ and let $\mathcal{U}$ be the uniform distribution. The algorithm updates $(W_N,V)$ as follows.
\\
\begin{tabular}[!]{|l|}
\hline\\
\parbox[!]{14cm}{
\texttt{
{\bf Algorithm 1:} $\Delta$-guided mixed pCN kernel\\
{\bf Input:} $(W_N,V)$
\begin{enumerate}\label{2ca}
\setlength\itemsep{-0.5em}
  \item Set $(\ddot{W}_N,\ddot{V})=(W_N,V)$
  \item While $(\Delta\ddot{W}_N-\Delta W_N)\times V\leq 0$
  \begin{enumerate}
  \item Simulate $H\sim Gamma(|N|/2,\Delta W_N/2)$
  \item Simulate $\ddot{W}_N\sim N(\mu_W+\sqrt{(1-\beta^2)}(W_N-\mu_W),\beta^2H^{-1}\Sigma_W)$
  \end{enumerate}
  \item Simulate $U\sim\mathcal{U}(0,1)$
  \item If $\ds U\leq \frac{\mathlarger{\pi}(\ddot{W}_N\mid \cdot) \lVert{W_N}\rVert^{-|N|}}{\mathlarger{\pi}(W_N\mid \cdot) \lVert{\ddot{W}_N}\rVert^{-|N|}}$, set $W_N=\ddot{W}_N$\\
  Else set $V=-V$
\end{enumerate}
{\bf Output:} $(W_N,V)$
}}\\  \hline
\end{tabular}\\

This algorithm is used for the dimension of $Z$ that introduces heavy tails, when $d>0$. The choices of \( \mu_W \) and \( \Sigma_W \) play a crucial role in the algorithm's efficiency. The optimal values of \( \mu_W \) and \( \Sigma_W \) are the mean and covariance of the target distribution; however, these are unavailable and the set \( N \) varies across the MCMC iterations. Our strategy is to use the mean and variance of \( (W_{N,j} \mid Y_{N,j}, \theta_f) \), and we set \( \Sigma_W \) to be diagonal. With suitable values of \( \mu_W \) and \( \Sigma_W \), \( \beta \in (0, 1) \) is tuned down to the value that either stabilizes the acceptance rate at a value higher than 0.234 or reaches a rate of 0.234.

For the reparameterization \( h \) of \( Z_N \), the asymmetry in the full conditional distribution arises from two factors: the asymmetry of the \( \text{Gamma}(\nu/2, \nu/2) \) prior on \( Z_{N,j} \) and the dependence structure in the distribution \( (Y_N \mid Z_N, N, \theta_f) \), which acts as a likelihood for \( Z_N \). Thus, we focus on reducing the asymmetry of the Gamma distribution when selecting the \( h \) function. In the Gamma distribution, asymmetry is influenced by the shape parameter: the smaller the shape parameter, the greater the asymmetry. Since \( \nu \) represents the degrees of freedom for the resulting Student-\( t \) distribution \( f \), we impose \( \nu > 2 \) to ensure finite variance. A Gamma distribution with a shape parameter of 1 represents our ``worst-case scenario" of asymmetry. Stabilization theory, as discussed by \citet{mcleod}, suggests that the transformation \( h(z) = (z^{\gamma} - 1) / \gamma \), for $\gamma>0$, yields a fairly symmetric distribution. For a shape parameter of 1, \(\gamma = 0.3\) provides effective symmetry, and this choice also works well for larger shape parameters. Therefore, we use \( W_{N,j} = (Z_{N,j}^{0.3} - 1) / 0.3 \) for \( j = 1, \ldots, |N| \).

Finally, the full conditional densities are given by:
\begin{eqnarray}
    \mathlarger{\pi}(W_N \mid \cdot) &\propto& \mathlarger{\pi}(Y_N \mid h^{-1}(W_N), N, \theta_{f}) \mathlarger{\pi}(W_N \mid N, \theta_{f,3}), \quad \text{if } d = 1, \label{ZNdistr2} \\
    \mathlarger{\pi}(W_N \mid \cdot) &\propto& \mathlarger{\pi}(Y_N \mid h^{-1}(W_N), Z_{N,2}, N, \theta_{f}) \mathlarger{\pi}(W_{N,1} \mid N, \theta_{f,3}), \quad \text{if } d = 2.
\end{eqnarray}

Given the \( \text{Gamma}(\nu/2, \nu/2) \) prior on \( Z_{N,j} \), we have:
\begin{equation}
    \mathlarger{\pi}(W_{N,j} \mid N, \theta_{f,3}) = \frac{(\nu/2)^{\nu/2}}{\gamma \Gamma(\nu/2)} (\gamma W_{N,j} + 1)^{2/\gamma + \nu/2 - 2} \exp\left(-\frac{\nu}{2} (\gamma W_{N,j} + 1)^{1/\gamma}\right) \mathbbm{1}[W_{N,j} > -\gamma^{-1}].
\end{equation}

The mean and variance of \( (W_{N,j} \mid Y_{N,j}, \theta_f) \) are derived from the Gamma distribution of \( (Z_{N,j} \mid Y_{N,j}, \theta_f) \) (see Appendix B). If this distribution is \( \text{Gamma}(a, b) \), then:
\begin{align}
    \mu_j &= \frac{\Gamma(a + \gamma)}{\Gamma(a)} \frac{b^{-\gamma}}{\gamma} - \frac{1}{\gamma}, \nonumber \\
    [\Sigma_W]_{jj} &= \frac{1}{\gamma^2} \left[\frac{\Gamma(a + 2\gamma)}{\Gamma(a)} \frac{b^{-2\gamma}}{\gamma} - \left(\frac{\Gamma(a + \gamma)}{\Gamma(a)} b^{-\gamma}\right)^2\right]. \nonumber
\end{align}

\subsection{Sampling $(N,Z)$}

Sampling $(N,Z)$ presents the most challenging step in the Gibbs sampling algorithm due to the variable dimensionality of $N$, which influences the dimensions of other components in the model. A straightforward approach is to update $N$ using a Metropolis-Hastings algorithm with the proposal distribution being the $PP(\lambda)$ prior. However, as the mean value of $|N|$ increases, this can lead to poor mixing. To mitigate this issue, we perform multiple updates in hypercubes (cubes or squares, if $S\subset\mathds{R}^3$ or $S\subset\mathds{R}^2$) with volume/area $\mu(S)/K$ (type 1) and $\mu(S)/(9K)$ (type 2), where $\mu(S)$ is the volume/area of $S$. In each iteration of the Gibbs sampling, around $K$ updates of type 1 squares should be performed. Without loss of generality, we consider the case $S\subset\mathds{R}^2$.

In each update, the centroid of the type 1 square is chosen uniformly is $S$, while the centroid of the type 2 square is selected uniformly from certain red-flagged squares with area $\mu(S)/K$. The criteria for identifying these red-flagged squares are outlined below.

Let $N_k$ and $Z_k$ denote the subsets of $N$ and $Z$ restricted to a given square $k$, type 1 or 2. Jointly sampling the component $Z_k$ along with $N_k$ enhances the efficiency of the algorithm and enables its validity.

The full conditional densities are the following:
\begin{eqnarray}
	\mathlarger{\pi}(N_k\mid \cdot) &\propto& \mathlarger{\pi}_{\mathcal{G}}(Y_u,Y_o\mid Y_N, N,\theta_{\mathcal{G}})\mathlarger{\pi}(Y_N\mid N,\theta_f)\mathlarger{\pi}(N\mid \lambda), \;\;\mbox{if } d=0, \\
	\mathlarger{\pi}(N_k,Z_k\mid \cdot) &\propto& \mathlarger{\pi}_{\mathcal{G}}(Y_u,Y_o\mid Y_N, N,\theta_{\mathcal{G}})\mathlarger{\pi}(Y_N\mid Z_N,N,\theta_f)\mathlarger{\pi}(Z_k\mid \theta_{f,3})\mathlarger{\pi}(N\mid \lambda), \;\;\mbox{if } d>0.
\end{eqnarray}

When $d=0$, the proposal distribution for $N_k$ is the prior Poisson process. For $d=1$, the proposal for $(N_k,Z_k)$ is given by:

\begin{equation}
	q(\ddot{N}_k,\ddot{Z}_k) = \mathlarger{\pi}(\ddot{N}_k\mid \lambda)q(\ddot{Z}_{\ddot{N}_k}\mid {Y}_{\ddot{N}_k},\ddot{N}_k,\theta_f)\mathlarger{\pi}(\ddot{Z}_{k\setminus\ddot{N}_k}\mid \theta_{f,3}),
\end{equation}

where $Z_{k\setminus\ddot{N}_k}$ refers to $Z_k\setminus\ddot{Z}_{\ddot{N}_k}$. The distribution $q(\ddot{Z}_{\ddot{N}_k}\mid {Y}_{\ddot{N}_k},\ddot{N}_k,\theta_f)$ is a calibrated version of:

\begin{equation}
	q^*(\ddot{Z}_{\ddot{N}_k}\mid {Y}_{\ddot{N}_k},\ddot{N}_k,\theta_f) = \prod_{j=1}^{\ddot{N}_k}\mathlarger{\pi}(\ddot{Z}_{\ddot{N}_k,j}\mid {Y}_{\ddot{N}_k,j},\ddot{N}_k,\theta_f).
\end{equation}

Specifically, $q$ adjusts the variance of $q^*$--typically reducing it--to achieve a reasonable acceptance rate. For instance, if $q^*$ is a Gamma distribution, $q$ would be defined with modified parameters, preserving the mean while calibrating the variance. The specific forms of $q$ for the distributions considered in this paper can be found in Appendix \ref{appendixB}.

The acceptance probabilities for the proposed updates are expressed as:

\begin{align}
	& 1 \land \frac{\mathlarger{\pi}_{\mathcal{G}}(Y_N\mid N,\theta_{\mathcal{G}})\mathlarger{\pi}(Y_{\ddot{N}}\mid \ddot{N},\theta_f)}{\mathlarger{\pi}_{\mathcal{G}}(Y_{\ddot{N}}\mid \ddot{N},\theta_{\mathcal{G}})\mathlarger{\pi}(Y_{N}\mid N,\theta_f)}, \;\;\mbox{if } d=0, \\
	& 1 \land \frac{\mathlarger{\pi}_{\mathcal{G}}(Y_N\mid N,\theta_{\mathcal{G}})\mathlarger{\pi}(Y_{\ddot{N}}\mid \ddot{Z}_{\ddot{N}},\ddot{N},\theta_f)\pi(\ddot{Z}_{\ddot{N}_k}\mid \ddot{N}_k,\theta_{f,3})q(Z_{N_k}\mid {Y}_{N_k},N_k,\theta_f)}{\mathlarger{\pi}_{\mathcal{G}}(Y_{\ddot{N}}\mid \ddot{N},\theta_{\mathcal{G}})\mathlarger{\pi}(Y_{N}\mid Z_{N},N,\theta_f)\pi(Z_{N_k}\mid N_k,\theta_{f,3})q(\ddot{Z}_{\ddot{N}_k}\mid {Y}_{\ddot{N}_k},\ddot{N}_k,\theta_f)}, \;\;\mbox{if } d>0.
\end{align}

This update step involves simulating $\ddot{N}_{k}$ from its prior, retrospectively sampling $Y$ at $\ddot{N}_{k}$ from $\mathlarger{\pi}_{\mathcal{G}}(Y_u\mid Y_N, Y_o, \theta_{\mathcal{G}})$, and, for cases where $d>0$, sampling $Z$ at $\ddot{N}_{k}$ from the distribution $q$ defined previously.

To establish the specific dynamics of the algorithm, we begin by tuning the values of $K$ and $c$ during a preliminary run, focusing solely on type 1 updates. A generally effective choice is $K \approx \lambda \mu(S)/5$. We select the value of $c$, among three or four fixed candidates (typically between 0.5 an 1.5), that yields the highest acceptance rate. During each iteration of the chain, we perform $K$ type 1 updates. Once $K$ and $c$ have been tuned, we run the chain for a specified duration and compute acceptance rates associated to the $K$ regular square covering $S$. We associate the acceptance indicator of a type 1 square with the regular square that minimizes the distance between the two centroids. Squares with acceptance rates below a designated threshold (e.g., 0.1) are flagged as red-flagged.

For each red-flagged square, we perform nine type 2 updates, with centroids chosen uniformly within the regular square. Subsequently, we calibrate a tuning parameter $c$ for the type 2 squares.

After updating each $(N_k,Z_k)$, a virtual step is executed to update $Y_u$ and $Z_u$ based on their respective full conditional distributions. This process essentially involves discarding any currently sampled finite set of $(Y_u,Z_u)$. These sets consist of $(Y,Z)$ at the previous $N_k$, upon acceptance, or on the rejected proposal. Finally, the infinite dimensionality of the MCMC chain, coupled with the retrospective sampling of $(Y_u,Z_u)$, removes the need for complex reversible jump steps.

\subsection{Sampling $\lambda$, $\theta_{\mathcal{G}}$, and $\theta_f$}

When $N$ is a homogeneous Poisson process with rate $\lambda$, standard conjugate Bayesian analysis reveals that, given a prior distribution of $Gamma(\alpha_{\lambda}, \beta_{\lambda})$, the full conditional distribution of $\lambda$ follows a $Gamma(\alpha_{\lambda} + |N|, \beta_{\lambda} + \mu(S))$.

In the case of a non-homogeneous PP characterized by a parametric intensity function $\lambda(s;\theta_{\lambda})$, the full conditional density of $\theta_{\lambda}$ is proportional to
\begin{align}
    \left[\exp\left\{-\int_S \lambda(s;\theta_{\lambda}) \, ds\right\} \prod_{j=1}^{|N|} \lambda(s_{N,j}; \theta_{\lambda})\right] \pi(\theta_{\lambda}), \label{acceptheta1}
\end{align}
where $\pi(\theta_{\lambda})$ denotes the prior density. In scenarios where a conjugate analysis is infeasible, $\theta_{\lambda}$ is updated using a Gaussian random walk MH step.

The parameters $\theta_{\mathcal{G}}$, $(\theta_{f,1}, \theta_{f,2})$, and $\theta_{f,3}$ are sampled in MH steps with appropriately tuned Gaussian random walk proposals that may be correlated \citep[see][]{roberts1997weak, roberts2009examples}. Moreover, $Y_u$ is jointly updated with $\theta_{\mathcal{G}}$ to ensure the algorithm's validity. The proposal distribution for this step is defined as
\begin{equation}
    q(\ddot{Y}_u, \ddot{\theta}_G \mid  \theta_{\mathcal{G}}) = \pi_{\mathcal{G}}(\ddot{Y}_u \mid  Y_o, Y_N, N, \ddot{\theta}_G) q(\ddot{\theta}_G \mid  \theta_{\mathcal{G}}).
\end{equation}

We can choose between a centered or non-centered parameterization for the missing data $Y_N$ and $Z_N$. According to \citet{PRSpar}, the non-centered parameterization is advantageous when the missing data is weakly identified by the data relative to the parameters. For datasets of reasonable size (e.g., $\geq \mathcal{O}(10^2)$), it is expected that $Y_N$ is strongly identified by the data, suggesting the use of a centered parameterization. Conversely, variables $Z_N$ are lower in the model hierarchy, and the distribution of $Y_{N,j}$ depends on $Z$ only through $Z_{N,j}$. This suggests that $Z_N$ is weakly identified by the data relative to $\theta_{f,3}$, motivating the adoption of a non-centered parameterization. This is also crucial for devising a valid MCMC, as it eliminates the dependence between $\theta_{f,3}$ and the reparameterized random field $Z$.

We define $\theta_{f,3,1}$ as the subset of parameters in $\theta_{f,3}$ that index the distribution of $(Y_N\mid Z_N)$, and $\theta_{f,3,2}$ as those that index the distribution of $Z$. Let $Z_{\theta}$ represent the dimension of $Z$ indexed by $\theta_{f,3,2}$, for which we adopt a non-centered parameterization by introducing i.i.d. random variables $Z_{\theta,N}^*(s) \sim U(0,1)$ for $s \in S$, and define 
\[
Z_{\theta,N}(s) = h(Z_{\theta,N}^*(s), \theta_{f,3,2}),
\]
where $h$ is the inverse c.d.f. of $Z_{\theta,N}(s)$. The $|N|$-vector $h(Z_{\theta,N}^*, \theta_{f,3,2})$ has entries $h(Z_{\theta,N,j}^*, \theta_{f,3,2})$.

The acceptance probabilities are expressed as follows:
\begin{align}
    \alpha((Y_u, \theta_{\mathcal{G}}), (\ddot{Y}_u, \ddot{\theta}_{\mathcal{G}})) &= 1 \land \frac{\pi_{\mathcal{G}}(Y_o \mid  Y_N, N, \ddot{\theta}_{\mathcal{G}}) \pi(\ddot{\theta}_{\mathcal{G}})}{\pi_{\mathcal{G}}(Y_o \mid  Y_N, N, \theta_{\mathcal{G}}) \pi(\theta_{\mathcal{G}})}, \nonumber \\
    \alpha((\theta_{f,1}, \theta_{f,2}), (\ddot{\theta}_{f,1}, \ddot{\theta}_{f,2})) &= 1 \land \frac{\pi(Y_N \mid  N, \ddot{\theta}_{f,1}, \ddot{\theta}_{f,2}, \theta_{f,3}) \pi(\ddot{\theta}_{f,1}, \ddot{\theta}_{f,2})}{\pi(Y_N \mid  N, \theta_{f,1}, \theta_{f,2}, \theta_{f,3}) \pi(\theta_{f,1}, \theta_{f,2})}, \quad \text{if } d=0, \nonumber \\
    \alpha((\theta_{f,1}, \theta_{f,2}), (\ddot{\theta}_{f,1}, \ddot{\theta}_{f,2})) &= 1 \land \frac{\pi(Y_N \mid  Z_N, N, \ddot{\theta}_{f,1}, \ddot{\theta}_{f,2}, \theta_{f,3}) \pi(\ddot{\theta}_{f,1}, \ddot{\theta}_{f,2})}{\pi(Y_N \mid  Z_N, N, \theta_{f,1}, \theta_{f,2}, \theta_{f,3}) \pi(\theta_{f,1}, \theta_{f,2})}, \quad \text{if } d>0, \nonumber \\
    \alpha(\theta_{f,3}, \ddot{\theta}_{f,3}) &= 1 \land \frac{\pi(Y_N \mid  N, \theta_{f,1}, \theta_{f,2}, \ddot{\theta}_{f,3}) \pi(\ddot{\theta}_{f,3})}{\pi(Y_N \mid  N, \theta_{f,1}, \theta_{f,2}, \theta_{f,3}) \pi(\theta_{f,3})}, \quad \text{if } d=0, \nonumber \\
    \alpha(\theta_{f,3,1}, \ddot{\theta}_{f,3,1}) &= 1 \land \frac{\pi(Y_N \mid  Z_N, N, \theta_{f,1}, \theta_{f,2}, \ddot{\theta}_{f,3,1}) \pi(\ddot{\theta}_{f,3,1})}{\pi(Y_N \mid  Z_N, N, \theta_{f,1}, \theta_{f,2}, \theta_{f,3,1}) \pi(\theta_{f,3,1})}, \quad \text{if } d>0, \nonumber \\
    \alpha(\theta_{f,3,2}, \ddot{\theta}_{f,3,2}) &= 1 \land \frac{\pi(Y_N \mid  h(Z_{\theta,N}^*, \ddot{\theta}_{f,3,2}), N, \theta_{f,1}, \theta_{f,2}, \theta_{f,3,1}) \pi(\ddot{\theta}_{f,3,2})}{\pi(Y_N \mid  h(Z_{\theta,N}^*, \theta_{f,3,2}), N, \theta_{f,1}, \theta_{f,2}, \theta_{f,3,1}) \pi(\theta_{f,3,2})}, \quad \text{if } d>0.
\end{align}

When the location parameters in $\theta_{\mathcal{G}}$ (for any $d$) and $\theta_f$ (for $d>0$) exhibit only slight correlations with the respective variance and correlation parameters, it is advisable to sample these separately, as they yield a normal full conditional distribution.

It is often beneficial to match the mean, variance, and correlation of the base GP with the $f$ distribution. In such cases, a non-centered parameterization of $Y_N$ is recommended, specially for smaller datasets. Parameters $\mu$ and $\sigma^2$ can be decoupled from the distribution of $Y_N$ by defining $\ds X_{N}=g(Y_N)=\frac{Y_N-\mu}{\sigma}$, for $d>0$, and by considering the $g$ function defined in Section \ref{FCYN}, for $d=0$.
Let $\phi$ denote the parameters that define the correlation structure of both the base GP and the distribution of $f$. The corresponding acceptance probability then becomes:
\begin{align}
    \alpha((Y_u, \theta_{\mathcal{G}}), (\ddot{Y}_u, \ddot{\theta}_{\mathcal{G}})) &= 1 \land \frac{\pi_{\mathcal{G}}(Y_o \mid  g^{-1}(X_{N},\ddot{\theta}_{\mathcal{G}}), N, \ddot{\theta}_{\mathcal{G}}) \pi(X_N \mid  N, \ddot{\phi}, \theta_{f,3}) \pi(\ddot{\theta}_{\mathcal{G}})}{\pi_{\mathcal{G}}(Y_o \mid  g^{-1}(X_{N},\theta_{\mathcal{G}}), N, \theta_{\mathcal{G}}) \pi(X_N \mid  N, \phi, \theta_{f,3}) \pi(\theta_{\mathcal{G}})}, \quad \text{if } d=0, \nonumber \\
\alpha((Y_u,\theta_{\mathcal{G}}),(\ddot{Y}_u,\ddot{\theta}_{\mathcal{G}})) &= 1 \land \cfrac{\mathlarger{\pi}_{\mathcal{G}}(Y_o\mid g^{-1}(X_{N},\ddot{\theta}_{\mathcal{G}}), N, \ddot{\theta}_{\mathcal{G}})\mathlarger{\pi}(X_{N}|Z_N,N,\ddot{\phi},\theta_{f,3})\mathlarger{\pi}(\ddot{\theta}_{\mathcal{G}})}{\mathlarger{\pi}_{\mathcal{G}}(Y_o\mid g^{-1}(X_{N},\theta_{\mathcal{G}}), N, \theta_{\mathcal{G}})\mathlarger{\pi}(X_{N}|Z_N,N,\phi,\theta_{f,3})\mathlarger{\pi}(\theta_{\mathcal{G}})},\;\;\mbox{if } d>0. \nonumber
\end{align}


The acceptance probabilities for \( \theta_{f,3} \) when \( d = 0 \), and for \( \theta_{f,3,1} \) and \( \theta_{f,3,2} \) when \( d > 0 \), can be obtained by straightforward adaptations of the expressions already provided.


\subsection{Prediction}

Under the MCMC approach presented above, it is straightforward to perform prediction for functions of the unobserved (infinite-dimensional) part $Y_u$ of the process $Y$.

Let $h(Y_u)$ be some finite-dimensional real and tractable function of $Y_u$. This includes, for example, the process $Y$ at a finite collection of locations. Under the Bayesian approach, prediction ought to be performed through the posterior predictive distribution of $h$, i.e., $\pi(h(Y_u)\mid Y_o)$. An approximate sample from this distribution can be obtained within the proposed MCMC algorithm by sampling from the respective full conditional distribution of $h(Y_u)$ at each iteration of the Gibbs sampler.
Function $h(Y_u)$ is simulated from its full conditional distribution by sampling $Y_u$, from $\pi_{\mathcal{G}}(Y_u|Y_N,N,Y_o,\theta_G)$, at the finite collection of locations required to compute $h(Y_u)$.

More theoretically advanced yet practically straightforward simulation techniques enable Monte Carlo estimation of certain intractable functions \( h(Y_u) \). For example, let \( h(Y_u) = \int_S \tilde{h}(Y(s)) \, ds \), where \( \tilde{h} \) is a tractable function. We define an i.i.d. sample \( U^{(1)}, \ldots, U^{(M)} \) from a uniform distribution over \( S \) and consider the following Monte Carlo estimator for \( E_{Y_u | Y_o}[h(Y_u)] \):
\begin{equation}\label{unb_est}
    \hat{h}(Y_u) = \frac{\mu(S)}{M} \sum_{j=1}^{M} \tilde{h}(Y^{(j)}(U^{(j)})),
\end{equation}
where the superscript \( (j) \) denotes the \( j \)-th value from the MCMC sample of size \( M \).

The justification for this estimator lies in the fact that
\begin{equation}
    E_U[\mu(S)\tilde{h}(Y(U))] = h(Y_u).
\end{equation}

To improve this estimator in terms of variance reduction, we divide \( S \) into \( K \) equal squares (cubes/hypercubes) and apply the same estimation procedure within each one. The final estimator is then obtained by summing the \( K \) individual estimators.

\bibliographystyle{apalike}
\bibliography{biblio}

\begin{thebibliography}{}

\bibitem[Alodat and Al-Rawwash, 2009]{alodat2009skew}
Alodat, M. and Al-Rawwash, M. (2009).
\newblock Skew-{G}aussian random field.
\newblock {\em Journal of computational and applied mathematics},
  232(2):496--504.

\bibitem[Alodat and AL-Rawwash, 2014]{alodat2014extended}
Alodat, M. and AL-Rawwash, M. (2014).
\newblock The extended skew {G}aussian process for regression.
\newblock {\em Metron}, 72(3):317--330.

\bibitem[B{\'a}rdossy, 2006]{bardossy2006copula}
B{\'a}rdossy, A. (2006).
\newblock Copula-based geostatistical models for groundwater quality
  parameters.
\newblock {\em Water Resources Research}, 42(11).

\bibitem[B{\'a}rdossy and Li, 2008]{bardossy2008geostatistical}
B{\'a}rdossy, A. and Li, J. (2008).
\newblock Geostatistical interpolation using copulas.
\newblock {\em Water resources research}, 44(7).

\bibitem[Beskos et~al., 2006]{beskos2006exact}
Beskos, A., Papaspiliopoulos, O., Roberts, G.~O., and Fearnhead, P. (2006).
\newblock Exact and computationally efficient likelihood-based estimation for
  discretely observed diffusion processes (with discussion).
\newblock {\em Journal of the Royal Statistical Society: Series B (Statistical
  Methodology)}, 68(3):333--382.

\bibitem[Beskos and Roberts, 2005]{beskos2005exact}
Beskos, A. and Roberts, G.~O. (2005).
\newblock Exact simulation of diffusions.
\newblock {\em The Annals of Applied Probability}, 15(4):2422--2444.

\bibitem[Bevilacqua et~al., 2021]{bevilacqua2021non}
Bevilacqua, M., Caama{\~n}o-Carrillo, C., Arellano-Valle, R.~B., and
  Morales-O{\~n}ate, V. (2021).
\newblock Non-{G}aussian geostatistical modeling using (skew) t processes.
\newblock {\em Scandinavian Journal of Statistics}, 48(1):212--245.

\bibitem[Bispo et~al., 2020]{bispo2020}
Bispo, N.~S., Prates, M.~O., and Gon\c{c}alves, F.~B. (2020).
\newblock {B}ayesian linear regression models with flexible error
  distributions.
\newblock {\em Journal of Statistical Computation and Simulation},
  90:2571--2591.

\bibitem[Cotter et~al., 2013]{cotter}
Cotter, S.~L., Roberts, G.~O., Stuart, A.~M., and White, D. (2013).
\newblock {MCMC} methods for functions: Modifying old algorithms to make them
  faster.
\newblock {\em Statistical Science}, 28:424--446.

\bibitem[Cressie, 2015]{cressie2015statistics}
Cressie, N. (2015).
\newblock {\em Statistics for spatial data}.
\newblock John Wiley \& Sons, Hoboken.

\bibitem[De~Oliveira et~al., 1997]{de1997bayesian}
De~Oliveira, V., Kedem, B., and Short, D.~A. (1997).
\newblock {B}ayesian prediction of transformed {G}aussian random fields.
\newblock {\em Journal of the American Statistical Association},
  92(440):1422--1433.

\bibitem[Dubrule, 1989]{dubrule1989review}
Dubrule, O. (1989).
\newblock A review of stochastic models for petroleum reservoirs.
\newblock {\em GeoStatistics}, pages 493--506.

\bibitem[Fuglstad et~al., 2019]{fuglstad2019constructing}
Fuglstad, G.-A., Simpson, D., Lindgren, F., and Rue, H. (2019).
\newblock Constructing priors that penalize the complexity of {G}aussian random
  fields.
\newblock {\em Journal of the American Statistical Association},
  114(525):445--452.

\bibitem[Genton and Zhang, 2012a]{Genton12}
Genton, M.~G. and Zhang, H. (2012a).
\newblock Identifiability problems in some non-{G}aussian spatial random
  fields.
\newblock {\em Chilean Journal of Statistics}, 3:171--179.

\bibitem[Genton and Zhang, 2012b]{genton2012identifiability}
Genton, M.~G. and Zhang, H. (2012b).
\newblock Identifiability problems in some non-{G}aussian spatial random
  fields.
\newblock {\em Chilean Journal of Statistics}, 3(2):171--179.

\bibitem[Gon{\c{c}}alves and Dias, 2023]{gonccalves2020exact}
Gon{\c{c}}alves, F.~B. and Dias, B. C.~C. (2023).
\newblock Exact {B}ayesian inference for level-set {C}ox processes.
\newblock {\em Journal of Computational and Graphical Statistics}, 32:1--18.

\bibitem[Gon{\c{c}}alves and Gamerman, 2018]{gonccalves2018exact}
Gon{\c{c}}alves, F.~B. and Gamerman, D. (2018).
\newblock Exact {B}ayesian inference in spatiotemporal {C}ox processes driven
  by multivariate {G}aussian processes.
\newblock {\em Journal of the Royal Statistical Society: Series B (Statistical
  Methodology)}, 80(1):157--175.

\bibitem[Gon{\c{c}}alves et~al., 2023]{gonccalves2023exact}
Gon{\c{c}}alves, F.~B., Latuszynski, K.~G., and Roberts, G.~O. (2023).
\newblock Exact {M}onte {C}arlo likelihood-based inference for jump-diffusion
  processes.
\newblock {\em Journal of the Royal Statistical Society - Series B},
  85:732--756.

\bibitem[Gr{\"a}ler et~al., 2010]{graler2010copulas}
Gr{\"a}ler, B., Kazianka, H., and de~Espindola, G.~M. (2010).
\newblock Copulas, a novel approach to model spatial and spatio-temporal
  dependence.
\newblock In {\em GIScience for Environmental Change Symposium Proceedings},
  volume~40, pages 49--54.

\bibitem[Gr{\"a}ler and Pebesma, 2011]{graler2011pair}
Gr{\"a}ler, B. and Pebesma, E. (2011).
\newblock The pair-copula construction for spatial data: a new approach to
  model spatial dependency.
\newblock {\em Procedia Environmental Sciences}, 7:206--211.

\bibitem[Hohn, 1998]{hohn1998geostatistics}
Hohn, M. (1998).
\newblock {\em GeoStatistics and petroleum geology}.
\newblock Springer Science \& Business Media.

\bibitem[Hughes, 2015]{hughes2015copcar}
Hughes, J. (2015).
\newblock copcar: A flexible regression model for areal data.
\newblock {\em Journal of Computational and Graphical Statistics},
  24(3):733--755.

\bibitem[Kamatani and Song, 2023]{kamatani}
Kamatani, K. and Song, X. (2023).
\newblock Non-reversible guided {M}etropolis kernel.
\newblock {\em Journal of Applied Probability}, 60:955--981.

\bibitem[Kazianka and Pilz, 2010]{kazianka2010copula}
Kazianka, H. and Pilz, J. (2010).
\newblock Copula-based geostatistical modeling of continuous and discrete data
  including covariates.
\newblock {\em Stochastic environmental research and risk assessment},
  24(5):661--673.

\bibitem[Kazianka and Pilz, 2011]{kazianka2011bayesian}
Kazianka, H. and Pilz, J. (2011).
\newblock {B}ayesian spatial modeling and interpolation using copulas.
\newblock {\em Computers \& Geosciences}, 37(3):310--319.

\bibitem[Mahmoudian, 2018]{mahmoudian2018existence}
Mahmoudian, B. (2018).
\newblock On the existence of some skew-{G}aussian random field models.
\newblock {\em Statistics \& Probability Letters}, 137:331--335.

\bibitem[Mcleod, 2013]{mcleod}
Mcleod, I. (2013).
\newblock Transformation to symmetry of gamma random variables.
\newblock Wolfram Demonstrations Project - A Wolfram Web Resource =
  \url{https://demonstrations.wolfram.com/TransformationToSymmetryOfGammaRandomVariables/}.

\bibitem[Palacios and Steel, 2006]{palacios2006non}
Palacios, M.~B. and Steel, M. F.~J. (2006).
\newblock Non-{G}aussian {B}ayesian geostatistical modeling.
\newblock {\em Journal of the American Statistical Association},
  101(474):604--618.

\bibitem[Papaspiliopoulos and Roberts, 2004]{papaspiliopoulos2004retrospective}
Papaspiliopoulos, O. and Roberts, G. (2004).
\newblock Retrospective mcmc methods for dirichlet process hierarchical models.
\newblock {\em Manuscript. Department of Mathematics and Statistics, Lancaster
  University. Submitted for publication, http://www. maths. lancs. ac. uk/\~{}
  papaspil/research. html. In revision}.

\bibitem[Papaspiliopoulos et~al., 2007]{PRSpar}
Papaspiliopoulos, O., Roberts, G.~O., and Sk\"{o}ld, M. (2007).
\newblock A general framework for the parametrization of hierarchical models.
\newblock {\em Statistical Science}, 22:59--73.

\bibitem[Prates et~al., 2012]{prates2012dengue}
Prates, M.~O., Dey, D.~K., and Lachos, V.~H. (2012).
\newblock A dengue fever study in the state of rio de janeiro with the use of
  generalized skew-normal/independent spatial fields.
\newblock {\em Chilean Journal of Statistics (ChJS)}, 3(2).

\bibitem[Prates et~al., 2015]{prates2015transformed}
Prates, M.~O., Dey, D.~K., Willig, M.~R., and Yan, J. (2015).
\newblock Transformed {G}aussian {M}arkov random fields and spatial modeling of
  species abundance.
\newblock {\em Spatial Statistics}, 14:382--399.

\bibitem[Resnick, 2013]{resnick2013probability}
Resnick, S.~I. (2013).
\newblock {\em A probability path}.
\newblock Springer Science \& Business Media.

\bibitem[Roberts et~al., 1997]{roberts1997weak}
Roberts, G.~O., Gelman, A., Gilks, W.~R., et~al. (1997).
\newblock Weak convergence and optimal scaling of random walk metropolis
  algorithms.
\newblock {\em The annals of applied probability}, 7(1):110--120.

\bibitem[Roberts and Rosenthal, 2009]{roberts2009examples}
Roberts, G.~O. and Rosenthal, J.~S. (2009).
\newblock Examples of adaptive mcmc.
\newblock {\em Journal of Computational and Graphical Statistics},
  18(2):349--367.

\bibitem[Rodr{\'i}guez-Yam et~al., 2004]{Yam2004}
Rodr{\'i}guez-Yam, G.~A., Davis, R.~A., and Scharf, L.~L. (2004).
\newblock Efficient gibbs sampling of truncated multivariate normal with
  application to constrained linear regression.

\bibitem[Schilling, 2005]{schill2005}
Schilling, R.~L. (2005).
\newblock {\em Measures, Integrals and Martingales}.
\newblock Cambridge University Press, Cambridge.

\bibitem[Shao, 2003]{shao2003mathematical}
Shao, J. (2003).
\newblock {\em Mathematical Statistics}.
\newblock Springer Texts in Statistics. Springer.

\bibitem[Simpson et~al., 2017]{penprior2017}
Simpson, D., Rue, H., Riebler, A., Martins, T.~G., and S{\o}rbye, S.~H. (2017).
\newblock Penalising model component complexity: A principled, practical
  approach to constructing priors.
\newblock {\em Statistical Science}, 32:1--28.

\bibitem[Tagle et~al., 2020]{tagle2020hierarchical}
Tagle, F., Castruccio, S., and Genton, M.~G. (2020).
\newblock A hierarchical bi-resolution spatial skew-t model.
\newblock {\em Spatial Statistics}, 35:100398.

\bibitem[Zhang, 2004]{Zhang04}
Zhang, H. (2004).
\newblock Inconsistent estimation and asymptotically equivalent interpolations
  in model-based geostatistics.
\newblock {\em Journal of the American Statistical Association}, 99:250--261.

\bibitem[Zhang and El-Shaarawi, 2010]{zhang2010spatial}
Zhang, H. and El-Shaarawi, A. (2010).
\newblock On spatial skew-{G}aussian processes and applications.
\newblock {\em Environmetrics: The official journal of the International
  Environmetrics Society}, 21(1):33--47.

\bibitem[Zheng et~al., 2021]{zheng2021nearest}
Zheng, X., Kottas, A., and Sans{\'o}, B. (2021).
\newblock Nearest-neighbor geostatistical models for non-{G}aussian data.
\newblock {\em arXiv preprint arXiv:2107.07736}.

\end{thebibliography}

\appendix

\section{Proofs}\label{appendixA}

\begin{proof}[Proof of Lemma \ref{lemma1}]
	To demonstrate the existence of the probability measure $\mu$, we verify that it satisfies the Kolmogorov Axioms.

	\textbf{(i) Non-negativity and boundedness}: For any $A \in \mathcal{F}$, we have:
	\begin{align}
		\mu(A) = \int_{\Omega_1}\int_{\Omega_2} \mathbbm{1}[(\omega_1, \omega_2) \in A] \, d\mu_{2,\omega_1}(\omega_2) \, d\mu_1(\omega_1).
	\end{align}
	Noting that $\mu_{2,\omega_1}$ is a probability measure on $(\Omega_2, \mathcal{F}_2)$, we have $0 \le g_A(\omega_1) \le 1$ for all $(\omega_1, A) \in (\Omega_1 \times \mathcal{F})$. Then,
	\begin{align}
		0 = \int_{\Omega_1} 0 \, d\mu_1(\omega_1) \le \mu(A) \le \int_{\Omega_1} 1 \, d\mu_1(\omega_1) = 1.
	\end{align}

	\textbf{(ii) Normalization}: We show that $\mu(\Omega) = 1$. For each $\omega_1 \in \Omega_1$,
	\begin{align}
		\int_{\Omega_2} \mathbbm{1}[(\omega_1, \omega_2) \in \Omega] \, d\mu_{2,\omega_1}(\omega_2) = 1.
	\end{align}
	Thus,
	\begin{align}
		\mu(\Omega) = \int_{\Omega_1} \int_{\Omega_2} \mathbbm{1}[(\omega_1, \omega_2) \in \Omega] \, d\mu_{2,\omega_1}(\omega_2) \, d\mu_1(\omega_1) = \int_{\Omega_1} 1 \, d\mu_1(\omega_1) = 1.
	\end{align}

	\textbf{(iii) Countable additivity ($\sigma$-additivity)}: For any countable sequence of disjoint events $\{A_i\}_{i=1}^{\infty} \subset \mathcal{F}$, we show that $\mu\left(\bigcup_{i=1}^{\infty} A_i\right) = \sum_{i=1}^{\infty} \mu(A_i)$. Define $A = \bigcup_{i=1}^{\infty} A_i$, and observe that $\mathbbm{1}[(\omega_1, \omega_2) \in A] = \sum_{i=1}^{\infty} \mathbbm{1}[(\omega_1, \omega_2) \in A_i]$. Then,
	\begin{align}
		\mu\left(\bigcup_{i=1}^{\infty} A_i\right) &= \mu(A) = \int_{\Omega_1} \int_{\Omega_2} \sum_{i=1}^{\infty} \mathbbm{1}[(\omega_1, \omega_2) \in A_i] \, d\mu_{2,\omega_1}(\omega_2) \, d\mu_1(\omega_1).
	\end{align}
	Since $d\mu_{2,\omega_1}(\omega_2)$ is a probability measure, we have
	\begin{align}
		\mu(A) = \int_{\Omega_1} \sum_{i=1}^{\infty} \left[\int_{\Omega_2} \mathbbm{1}[(\omega_1, \omega_2) \in A_i] \, d\mu_{2,\omega_1}(\omega_2)\right] d\mu_1(\omega_1).
	\end{align}
	By assumption, each $g_{A_i}(\omega_1)$ is $\mathcal{F}_1$-measurable, so
	\begin{align}
		\mu(A) &= \int_{\Omega_1} \sum_{i=1}^{\infty} g_{A_i}(\omega_1) \, d\mu_1(\omega_1) = \sum_{i=1}^{\infty} \int_{\Omega_1} g_{A_i}(\omega_1) \, d\mu_1(\omega_1) = \sum_{i=1}^{\infty} \mu(A_i).
	\end{align}
	Therefore, $\mu$ is a probability measure.
\end{proof}

\begin{proof}[Proof of Theorem \ref{theoexistence}]
	\textbf{(a) Existence of the POGAMP for a fixed $N$.}

Recall that $\mathcal{G}_Y$ is the probability measure of the GP mean function \( \mu \) and covariance function \( \Sigma \). Let $\mathcal{G}_N$ denote the measure of $Y_N$ under $\mathcal{G}_Y$ for fixed $N$, and $\mathcal{P}_N$ the probability measure of $Y_N$ induced by $f$. Define $\nu$ as the Lebesgue measure, and set $\ds g = \frac{d\mathcal{G}_N}{d\nu^{|N|}}$.
	
Next, define functions $h_1\colon \mathcal{Y}\to \mathcal{X}^{|N|}$, $h_2\colon \mathcal{X}^{|N|}\to \mathds{R}^+$, and $h\colon \mathcal{Y} \to \mathds{R}^+$ such that:
$h_1(y) = y_N$, $h_2(y_N) = \frac{f}{g}(y_N)$, $h(y) = h_2 \circ h_1(y)$.

Since $h_1$ is $\mathcal{B}(\mathcal{Y})$-measurable (by definition) and $h_2$ is $\mathcal{B}(\mathcal{X}^{|N|})$-measurable (by the continuity of $h_2$), $h$ is also $\mathcal{B}(\mathcal{Y})$-measurable. Therefore, by the Radon-Nikodym Theorem, there exists a $\sigma$-finite measure $\mathcal{P}_Y$ on $(\mathcal{Y}, \mathcal{B}(\mathcal{Y}))$ given by
\begin{align}
	\mathcal{P}_Y(A) = \int_A h(y) \, d\mathcal{G}_Y(y), \quad \forall A \in \mathcal{B}(\mathcal{Y}),
\end{align}
such that
\begin{align}
	\frac{d\mathcal{P}_Y}{d\mathcal{G}_Y}(y) = h(y).
\end{align}

Now, we confirm that $\mathcal{P}_Y$ is indeed a probability measure.

(i) Non-negativity and boundedness: $0 \le \mathcal{P}_Y(A) \le 1$, $\forall A \in \mathcal{B}(\mathcal{Y})$.
\begin{align}
	\mathcal{P}_Y(A) &= \int_A h(y) \, d\mathcal{G}_Y(y) \ge 0,\\
	\mathcal{P}_Y(A) &\le \int_{\mathcal{Y}} h(y) \, d\mathcal{G}_Y(y) = \int_{\mathcal{X}^N} \frac{f}{g}(y_N) g(y_N) \, d\nu = \int_{\mathcal{X}^N} f(y_N) \, d\nu = 1.
\end{align}

(ii) Normalization: $\mathcal{P}_Y(\mathcal{Y}) = 1$, as shown in the second part of (i).

Having established that $\mathcal{P}_Y$ is a probability measure, we proceed to demonstrate that it is indeed the POGAMP measure defined in Definition \ref{defexistY} for fixed $N$.

Define $Y_{N^c}$ to represent $Y$ at $S \setminus S_N$, and let $\mathcal{P}_{N^c}$ and $\mathcal{G}_{N^c}$ be the conditional probability measures of $(Y_{N^c} \mid Y_N)$ under $\mathcal{P}_Y$ and $\mathcal{G}_Y$, respectively. Since $\mathcal{P}_N \ll \mathcal{G}_N \ll \nu^{|N|}$, we obtain
\begin{align}
	\label{eq1teo1}
	\frac{f}{g}(y_N) &= \frac{d\mathcal{P}_Y}{d\mathcal{G}_Y}(y) =  \frac{d\mathcal{P}_N}{d\mathcal{G}_N}(y_N) \cdot \frac{d\mathcal{P}_{N^c}}{d\mathcal{G}_{N^c}}(y_{N^c}) \\
	&= \frac{\frac{d\mathcal{P}_N}{d\nu^{|N|}}}{\frac{d\mathcal{G}_N}{d\nu^{|N|}}}(y_N) \cdot \frac{d\mathcal{P}_{N^c}}{d\mathcal{G}_{N^c}}(y_{N^c}) \quad \text{$\mathcal{G}_Y$-a.s.},
\end{align}
which implies that there exists a constant $c \in \mathbb{R}^+$ such that $\frac{d\mathcal{P}_N}{d\nu^{|N|}} = \frac{f}{c}$ $\mathcal{G}_Y$-a.s. Since $\int_{\mathbb{R}^N} \frac{d\mathcal{P}_N}{d\nu^{|N|}}(y_N) \, d\nu = \int_{\mathbb{R}^N} f(y_N) \, d\nu = 1$, we have that $\frac{d\mathcal{P}_N}{d\nu^{|N|}} = f$ and $\frac{d\mathcal{P}_{N^c}}{d\mathcal{G}_{N^c}}(y_{N^c})=1$, $\mathcal{G}_Y$-a.s. This verifies that $\mathcal{P}_Y$ satisfies conditions (ii) and (iii) of Definition \ref{defexistY}, proving that it is the POGAMP measure for fixed $N$.
	
(b) \textbf{Case where $N$ is not fixed.}

To handle the case where $N$ is not fixed, we decompose $N$ into $(|N|, S_N)$, where $|N|$ represents the number of elements in $N$ and $S_N$ the locations themselves. Both are $\mathcal{B}(\mathcal{N})$-measurable, with $|N|$ and $S_N$ defining the measurable spaces $(\mathds{N}, \mathcal{B}(\mathds{N}))$ and $(S^{|N|}, \mathcal{B}(S^{|N|}))$, respectively. For any fixed $|N|$, if the integral
\begin{equation}\label{mensint}
g_A(S_N) = \int_{\mathcal{Y}} \mathbbm{1}[(S_N, y) \in A] \, d\mathcal{P}_Y
\end{equation}
is $\mathcal{B}(S^{|N|})$-measurable for all $A \in \sigma(\mathcal{B}(S^{|N|}) \times \mathcal{B}(\mathcal{Y}))$, Lemma \ref{lemma1} implies the existence of the POGAMP joint measure of $(Y, S_N)$ for all fixed $|N|$.

For any $A = (A_{S_N} \times A_Y) \in \sigma(\mathcal{B}(S^{|N|}) \times \mathcal{B}(\mathcal{Y}))$, we have:
\begin{align}\label{mensint2}
g_A(S_N) &= \mathbbm{1}[S_N \in A_{S_N}] \int_{\mathcal{Y}} \mathbbm{1}[y \in A_Y] \frac{f}{g}(y_N) \, d\mathcal{G}_Y \\
&= \mathbbm{1}[S_N \in A_{S_N}] \int_{\mathcal{Y}} \mathbbm{1}[y \in A_Y] \frac{f}{g}(y_N) \, d\mathcal{G}_Y,
\end{align}
where $\mathbbm{1}[y \in A_Y] \frac{f}{g}(y_N) \leq \frac{f}{g}(y_N)$, and $\frac{f}{g}(y_N)$ is $\mathcal{G}_Y$-integrable for all $S_N$, as $\int_{\mathcal{Y}} \frac{f}{g}(y_N) \, d\mathcal{G}_Y = \int_{\mathcal{X}^{|N|}} \frac{f}{g}(y_N) g(y_N) \, dy_N = 1$.

Notice that $\mathbbm{1}[y \in A_Y]$ does not depend on $S_N$, and, given the continuity of $f/g$ in $y_N$ and the $\mathcal{G}_Y$-a.s. continuity of $Y$, the ratio $\frac{f}{g}(y_N)$ is $\mathcal{G}_Y$-almost-surely continuous in $S_N$. This implies the $\mathcal{G}_Y$-a.s. continuity of \( g_{A,3}(y, S_N) := \mathbbm{1}[y \in A_Y] \frac{f}{g}(y_N) \) in \( S_N \), and consequently, \( g_{A,3}(Y, s_{N,n}) \xrightarrow{\enskip a.s. \enskip} g_{A,3}(Y, s_N) \) for any sequence \( \{s_{N,n}\}_{n=1}^{\infty} \) with \( s_{N,n} \in S^{|N|} \) converging to \( s_N \in S^{|N|} \).

By the dominated convergence theorem, it follows that \( g_{A,2}(s_{N,n}) := \int_{\mathcal{Y}} g_{A,3}(y, s_N) \, \text{d}\mathcal{G}_Y \) converges to \( g_{A,2}(s_N) \). This implies that \( g_{A,2}(S_N) \) is a continuous function of \( S_N \) and is thus $\mathcal{B}(S^{|N|})$-measurable.

Finally, since \( g_{A,1}(S_N) := \mathbbm{1}[S_N \in A_{S_N}] \) is also $\mathcal{B}(S^{|N|})$-measurable, the product
\[
g_A(S_N) := g_{A,1}(S_N) g_{A,2}(S_N)
\]
is $\mathcal{B}(S^{|N|})$-measurable as well, thereby establishing the existence of the POGAMP measure for any fixed \( |N| \in \mathds{N} \).

To complete the proof, we need to show that the joint process \( (Y, N) \) exists when \( |N| \) is not fixed. Using Lemma \ref{lemma1}, it is sufficient to prove that, for any \( A = A_Y \times A_{N} \in \sigma(\mathcal{B}(\mathcal{Y})\times\mathcal{B}(\mathcal{N}))\), the following integral is $\mathcal{B}(\mathds{N})$-measurable:
\begin{align}
	\label{mensint2}
	g_A(|N|) = \int_{S^{|N|} \times \mathcal{Y}} \mathbbm{1}[(|N|, s_N) \in A_{N}] \, \mathbbm{1}[y \in A_Y] \left( \prod_{j=1}^{|N|} \frac{\lambda(s_{N,j})}{\Lambda_S} \right) \, \text{d}\mathcal{P}_{Y} \, \text{d}s_N,
\end{align}
where \( \Lambda_S = \int_S \lambda(s) \, \text{d}s \) and \( s_N = (s_{N,1}, \ldots, s_{N,|N|}) \). Since \( N \) has a discrete marginal measure, it suffices to show that \( g_A(|N|) \) is a real-valued function of \( |N| \).

Then,
\begin{align}
	g_A(|N|) &= \int_{S^{|N|}} \mathbbm{1}[(|N|, s_N) \in A_{N}] \left( \prod_{j=1}^{|N|} \frac{\lambda(s_{N,j})}{\Lambda_S} \right) \left( \int_{\mathcal{Y}} \mathbbm{1}[y \in A_Y]  \, \text{d}\mathcal{P}_Y \right) \, \text{d}s_N \\
	&\leq \int_{S^{|N|}} \left( \prod_{j=1}^{|N|} \frac{\lambda(s_{N,j})}{\Lambda_S} \right) \left( \int_{\mathcal{Y}} \, \text{d}\mathcal{P}_Y \right) \, \text{d}s_N = 1.
\end{align}

\end{proof}

\begin{proof}[Proof of Proposition \ref{prop1}]
    The result follows directly from part (a) of the proof of Theorem \ref{theoexistence}. Specifically, we have
    \begin{align}
        \frac{\text{d}\mathcal{P}}{\text{d}\mathcal{G}}(N, Y) = \frac{\text{d}\mathcal{P}}{\text{d}\mathcal{G}}(N) \cdot \frac{\text{d}\mathcal{P}}{\text{d}\mathcal{G}}(Y \mid N) = 1 \cdot \frac{f}{g}(Y_N) = \frac{f}{g}(Y_N).
    \end{align}
\end{proof}

\begin{proof}[Proof of Proposition \ref{prop2}]
    By the definition of the Kullback-Leibler divergence, we have
    \begin{align}
        D_{\text{KL}}(\mathcal{P} \parallel \mathcal{G}) &= \int \log \left(\frac{\text{d}\mathcal{P}}{\text{d}\mathcal{G}}\right) \, \text{d}\mathcal{P} \nonumber \\
        &= E_N \left[ E_{S_N} \left[ E_{Y_N} \left[\log \left(\frac{f}{g}(Y_N)\right) \mid S_N\right] \mid N \right]  \right] \nonumber \\
        &= \sum_{N=0}^{\infty} \frac{\text{e}^{-\Lambda_S} \Lambda_S^N}{N!} \int_{S^{|N|}} \left[ \int_{\mathcal{X}^{|N|}} \log\left( \frac{f}{g}(y_N)\right) f(y_N) \, \text{d}y_N \right] \prod_{j=1}^{|N|} \frac{\lambda(s_{N,j})}{\Lambda_S} \, \text{d}s_N.
    \end{align}
\end{proof}

\begin{proof}[Proof of Theorem \ref{theo2}]
    For each $n \in \mathds{N}$, let $\dot{S}_n$ denote the $r$-dimensional vector where the $j$-th entry, $\dot{s}_{n,j}$, is the location in the vector $S_N$ of the $n$-th POGAMP that is closest to the $j$-th entry $s_{R,j}$ in $S_R$. Define $Y_n$ as the component $Y$ of the $n$-th POGAMP, with $Y_{n,\dot{S}_n}$, $Y_{n,S_N}$, and $Y_{n,R}$ representing $Y_n$ at $\dot{S}_n$, $S_N$, and $S_R$, respectively.
    
    Now, let $\|\cdot\|$ denote the $L_2$-norm and define the event
    \[
    A_n = \left(\|\dot{s}_{n,1} - s_{R,1}\| > \epsilon \cap \cdots \cap \|\dot{s}_{n,r} - s_{R,r}\| > \epsilon\right).
    \]
    For all $\epsilon > 0$ such that the $r$ balls of radius $\epsilon$ around the points in $S_R$ are disjoint, we obtain
    \begin{equation} \label{pth2_1}
    P(A_n) \leq \exp(-r \pi \epsilon^2 \lambda_n),
    \end{equation}
    which implies that
    \begin{equation} \label{pth2_2}
    \sum_{n=1}^{\infty} P(A_n) < \infty.
    \end{equation}
    Therefore, $\dot{S}_n \xrightarrow{\enskip a.s. \enskip} S_R$ \citep[see][Theorem 1.8 (v)]{shao2003mathematical}.
    
    Moreover, by the almost sure continuity of the mapping $s \mapsto Y(s)$ under the POGAMP measure (see Proposition \ref{prop1}), the convergence $\dot{S}_n \xrightarrow{\enskip a.s. \enskip} S_R$ implies that $(Y_{n,R} - Y_{n,\dot{S}_n}) \xrightarrow{\enskip a.s. \enskip} 0$.
    
    Define $F_{n,\dot{S}_n}$ as the c.d.f. of $Y_{n,\dot{S}_n}$, and let $F_{\dot{S}_n}$ and $F_R$ denote the c.d.f. of the distribution $f$ at $\dot{S}_n$ and $S_R$, respectively. Then we have
    \begin{equation} \label{pth2_3}
    F_{n,\dot{S}_n}(y) = E_{\dot{S}_n} \left[F_{\dot{S}_n}(y)\right].
    \end{equation}
    Since $f_{\dot{S}_n}(y)$ is continuous in $(\dot{S}_n, y)$, it follows that $F_{\dot{S}_n}(y)$ is also continuous (by the uniform integrability of $f$ and Theorem 16.6 in \citet{schill2005}). The convergence $\dot{S}_n \xrightarrow{\enskip a.s. \enskip} S_R$ implies that $F_{\dot{S}_n}(y) \xrightarrow{\enskip a.s. \enskip} F_R(y)$ for all $y \in \mathcal{X}^r$. By the dominated convergence theorem, we conclude
    \begin{equation} \label{pth2_4}
    \lim_{n \to \infty} F_{n,\dot{S}_n}(y) = E_{\dot{S}_n} \left[ \lim_{n \to \infty} F_{\dot{S}_n}(y) \right] = F_R(y), \quad \forall y \in \mathcal{X}^r,
    \end{equation}
    which shows that $Y_{n,\dot{S}_n}$ converges in distribution to $F_R$.
    
    Finally, since $(Y_{n,R} - Y_{n,\dot{S}_n}) \xrightarrow{\enskip a.s. \enskip} 0$ and $Y_{n,\dot{S}_n} \xrightarrow{\enskip d \enskip} F_R$, Slutsky's theorem \citep[see][Theorem 8.6.1 (a)]{resnick2013probability} implies that $Y_{n,R} \xrightarrow{\enskip d \enskip} F_R$.
\end{proof}

\begin{proof}[Proof of Proposition \ref{prop4}]

    \begin{align}
        Cov(Y_1, Y_2) &= E[Cov(Y_1, Y_2 \mid  Y_N)] + Cov[E(Y_1 \mid  Y_N), E(Y_2 \mid  Y_N)] \nonumber \\
        &= E[Cov(Y_1, Y_2 \mid  Y_N)] \nonumber \\
        &\quad + Cov\left[\bm{\mu}_1 + \Sigma_{1}\Sigma_{N}^{-1}(Y_N - \bm{\mu}_N), \bm{\mu}_2 + \Sigma_{2}\Sigma_{N}^{-1}(Y_N - \bm{\mu}_N)\right] \nonumber \\
        &= E[Cov(Y_1, Y_2 \mid Y_N)] + Cov\left[\Sigma_{1}\Sigma_{N}^{-1}(Y_N - \bm{\mu}_N), \Sigma_{2}\Sigma_{N}^{-1}(Y_N - \bm{\mu}_N)\right] \nonumber \\
        &= E[Cov(Y_1, Y_2 \mid Y_N)] + E\left[\Sigma_{1}\Sigma_{N}^{-1} Cov(Y_N, Y_N \mid N) (\Sigma_{2}\Sigma_{N}^{-1})^\intercal\right] \nonumber \\
        &\quad + Cov\left[E(\Sigma_{1}\Sigma_{N}^{-1}(Y_N - \bm{\mu}_N) \mid N), E(\Sigma_{2}\Sigma_{N}^{-1}(Y_N - \bm{\mu}_N) \mid N)\right] \nonumber \\
        &= E[Cov(Y_1, Y_2 \mid Y_N)] + E\left[\Sigma_{1}\Sigma_{N}^{-1} \Sigma_{N,f} (\Sigma_{2}\Sigma_{N}^{-1})^\intercal\right]. \nonumber
    \end{align}

    If the covariance function is identical under the base GP and $f$, then
    \begin{align}
        Cov(Y_1, Y_2) &= E[Cov(Y_1, Y_2 \mid Y_N)] + E[\Sigma_{1}\Sigma_{N}^{-1} \Sigma_{2}^\intercal] \nonumber \\
        &= E\left[\rho(s_1, s_2) - \sum_{j=1}^N \sum_{i=1}^N \rho(s_1, s_{N,i}) \rho(s_{N,i}, s_{N,j})^{-1} \rho(s_2, s_{N,j})\right] \nonumber \\
        &\quad + E\left[\sum_{j=1}^N \sum_{i=1}^N \rho(s_1, s_{N,i}) \rho(s_{N,i}, s_{N,j})^{-1} \rho(s_2, s_{N,j})\right] \nonumber \\
        &= \rho(s_1, s_2),
    \end{align}
    where $\rho(s_1, s_2)$ represents the covariance of $(Y(s_1), Y(s_2))$, as required.
\end{proof}

\begin{proof}[Proof of Proposition \ref{prop5}]
    For a symmetric pair $(\mathbf{s}, \mathbf{s}')$ such that $\mathbf{s}' = g(\mathbf{s})$ for some rotation $g$, we have
    \begin{align}
        \pi(Y_{\mathbf{s}} \mid \theta) &= \int \pi_{\mathcal{G}}(Y_{\mathbf{s}} \mid y_n, n, \theta) f_N(y_n \mid n, \theta) \pi(n \mid \theta) \, d\mathcal{P}(y_n, n) \\
        &= \int \pi_{\mathcal{G}}(Y_{g(\mathbf{s})} \mid y_{g(n)}, g(n), \theta) f_N(y_{g(n)} \mid g(n), \theta) \pi(g(n) \mid \theta) \, d\mathcal{P}(y_{g(n)}, g(n)) \\
        &= \pi(Y_{g(\mathbf{s})} \mid \theta) = \pi(Y_{\mathbf{s}'} \mid \theta).
    \end{align}
    This establishes the required symmetry property.
\end{proof}


\section{Examples of $f$ distributions}\label{appendixB}

Consider the base GP with constant mean $\mu$, variance $\sigma^2$, and an isotropic covariance function $\Sigma$. Define $\Sigma_{oo}$ and $\Sigma_{NN}$ as the respective covariance matrices of $Y_o$, $Y_N$, and $\Sigma_{oN}$ as the cross-covariance matrix between $Y_o$ and $Y_N$.

\begin{enumerate}

\item The Multivariate $t$ Distribution

Let $Z$ be a random field of i.i.d. $Gamma(\nu/2,\nu/2)$ random variables for $\nu>0$. For parameters $(\xi,\tau^2)\in\mathds{R}\times\mathds{R}^+$ and a correlation matrix $\Upsilon$ of size $N\times N$, define
\begin{equation}
Y_N\mid Z_N \sim N\left( \xi_N, \; \tau^2 \operatorname{diag}\{Z_{N}^{-1/2}\}\Upsilon\operatorname{diag}\{Z_{N}^{-1/2}\}\right), \label{YN_t}
\end{equation}
where $\xi_N$ is the $|N|$-vector with all entries equal to $\xi$. The marginal distribution of $Y_N$ is a multivariate Student-$t$ distribution with mean $\xi$, scale matrix $\tau^2\Upsilon$, and degrees of freedom $\nu$.

The full conditional distribution of $Y_N$ is multivariate normal, with mean vector $\mu^*$ and covariance matrix $\Sigma^*$ given by:
\begin{equation}
\Sigma^* = \left( \Sigma_{NN}^{-1} \Sigma_{No} \Sigma_{2}^{-1} \Sigma_{oN} \Sigma_{NN}^{-1} + \Lambda^{-1} \right)^{-1}, \;\;\;
\mu^* = \Sigma^*\left(\Sigma_{NN}^{-1} \Sigma_{No} \Sigma_{2}^{-1} \tilde{\mu} + \Lambda^{-1} \xi_N \right),
\end{equation}
where $\Sigma_2 = \Sigma_{oo} - \Sigma_{oN} \Sigma_{NN}^{-1} \Sigma_{No}$, $\Lambda = \tau^2 \operatorname{diag}\{Z_{N}^{-1/2}\}\Upsilon\operatorname{diag}\{Z_{N}^{-1/2}\}$, and $\tilde{\mu} = Y_o - \mu + \Sigma_{oN} \Sigma_{NN}^{-1} \mu_N$, with $\mu_N$ as the $|N|$-vector of entries $\mu$.

The $Z_N$ block is sampled using a $\Delta$-guided mixed pCN kernel. The proposal distribution, denoted $q(\ddot{Z}_{\ddot{N}_k,j}\mid {Y}_{\ddot{N}_k,j},\ddot{N}_k,\theta_f)$, is a
$$\ds Gamma\left( c \frac{\nu+1}{2}, c \frac{\left(\frac{Y_{\ddot{N}_k,j}-\xi}{\tau}\right)^2 + \nu}{2} \right).$$

To align the base GP and the distribution $f$ have matching means and variances, we set $\xi = \mu$ and $\tau^2 = \varsigma^{-1} \sigma^2$, where $\varsigma = \frac{\nu}{\nu-2}$.
The distribution of $X_{N}\mid Z_N$ is obtained by making $\xi=0$ and $\tau^2=\varsigma^{-1}$ in \eqref{YN_t}.

\item The Multivariate Skew-Normal Distribution

Let $Z$ be a random field of i.i.d. standard half-normal random variables. For $(\xi, \tau^2) \in \mathds{R} \times \mathds{R}^+$, $\delta \in (-1,1)$, and a correlation matrix $\Upsilon$ of size $N \times N$, define:
\begin{equation}
Y_N\mid Z_N \sim N\left( \xi_N + \tau \delta Z_N, \; \tau^2(1-\delta^2)^{-1}\Upsilon\right), \label{YN_SN}
\end{equation}
resulting in a marginal distribution of $Y_N$ as a multivariate skew-normal with location parameter $\xi$, scale matrix $\tau^2\Upsilon$, and skewness parameter $\delta$.

The full conditional distribution of $Y_N$ is multivariate normal, with:
\begin{equation}
\Sigma^* = \left( \Sigma_{NN}^{-1} \Sigma_{No} \Sigma_{2}^{-1} \Sigma_{oN} \Sigma_{NN}^{-1} + \Lambda^{-1} \right)^{-1}, \;\;\;
\mu^* = \Sigma^*\left(\Sigma_{NN}^{-1} \Sigma_{No} \Sigma_{2}^{-1} \tilde{\mu} + \Lambda^{-1} \tilde{\xi}_N \right),
\end{equation}
where $\Sigma_2 = \Sigma_{oo} - \Sigma_{oN} \Sigma_{NN}^{-1} \Sigma_{No}$, $\Lambda = \tau^2(1-\delta^2)^{-1}\Upsilon$, $\tilde{\xi}_N = \xi_N + \tau \delta Z_N$, and $\tilde{\mu} = Y_o - \mu + \Sigma_{oN} \Sigma_{NN}^{-1} \mu_N$.

The full conditional distribution of \( Z_N \) is a truncated multivariate normal, restricted to \( (\mathds{R}^+)^{|N|} \), with mean vector \( \mu^* \) and covariance matrix \( \Sigma^* \) defined as
\begin{equation}
\Sigma^* = \left(\delta^2 (1 - \delta^2) \Upsilon^{-1} + I_N \right)^{-1}, \qquad \mu^* = \frac{\delta (1 - \delta^2)}{\tau} \Sigma^* \Upsilon^{-1} (Y_N - \xi_N).
\end{equation}

To sample from this distribution, we employ an embedded Gibbs sampler with blocks defined by the coordinates of \( Z_{N}^* \sim N(0, I_N) \), where \( Z_N = \mu^* + (\Sigma^*)^{1/2} Z_{N}^* \), and \( (\Sigma^*)^{1/2} \) denotes the lower triangular Cholesky decomposition of \( \Sigma^* \) \citep[see][]{Yam2004}. The constraints on \( Z_N \) translate into linear constraints on \( Z_{N}^* \), namely \( (\Sigma^*)^{1/2} Z_{N}^* > -\mu^* \). Consequently, each block \( Z_{N,j}^* \) follows a truncated standard normal distribution. The triangular structure of \( (\Sigma^*)^{1/2} \) enhances the efficiency of the embedded Gibbs sampler, as it allows the chain to start within the truncation region, leading to rapid convergence (approximately 20 iterations based on numerical experiments). The final sampled value is retained as a realization from the target full conditional distribution.

The proposal distribution \( q(\ddot{Z}_{\ddot{N}_k,j} | Y_{\ddot{N}_k,j}, \ddot{N}_k, \theta_f) \) is a normal distribution
\[
N\left( \frac{\delta (1 - \delta^2) \left( \frac{Y_{\ddot{N}_k,j} - \xi}{\tau} \right)}{1 + \delta^2 (1 - \delta^2)}, \; c \frac{1}{1 + \delta^2 (1 - \delta^2)} \right)
\]
truncated to the positive domain.

To ensure that the base GP and the distribution \( f \) have matching means and variances, we set \( \tau^2 = \varsigma^{-1} \sigma^2 \) and \( \xi = \mu - \eta \), with \( \varsigma = \left( \delta^2 (1 - 2/\pi) + (1 - \delta^2)^{-1} \right) \) and \( \eta = \tau \delta \sqrt{2/\pi} \). The distribution of $X_{N}\mid Z_N$ is obtained by making $\xi=- \eta$ and $\tau^2=\varsigma^{-1}$ in \eqref{YN_SN}.

\item The Multivariate Skew-$t$ Distribution

Let \( Z_1 \) be a random field of i.i.d. \( \text{Gamma}(\nu/2, \nu/2) \) random variables and \( Z_2 \) a random field of i.i.d. standard half-normal random variables, with \( \nu > 0 \). For \( (\xi, \tau^2) \in \mathds{R} \times \mathds{R}^+ \), \( \delta \in (-1,1) \), and a correlation matrix $\Upsilon$ of size $N \times N$, define
\begin{equation}
Y_N | Z_{N,1}, Z_{N,2} \sim N\left( \xi_N + \tau \delta \operatorname{diag}\{ Z_{N,1}^{-1/2} \} Z_{N,2} , \; \tau^2 (1 - \delta^2)^{-1} \operatorname{diag}\{ Z_{N,1}^{-1/2} \} \Upsilon \operatorname{diag}\{ Z_{N,1}^{-1/2} \} \right). \label{YN_St}
\end{equation}
Then, the marginal distribution of \( Y_N \) is a multivariate skew-\( t \) distribution with location parameter \( \xi \), scale matrix \( \tau^2 \Upsilon \), skewness parameter \( \delta \), and degrees of freedom \( \nu \).

The full conditional distribution of \( Y_N \) is a multivariate normal with mean vector \( \mu^* \) and covariance matrix \( \Sigma^* \), given by
\begin{equation}
\Sigma^* = \left( \Sigma_{NN}^{-1} \Sigma_{No} \Sigma_{2}^{-1} \Sigma_{oN} \Sigma_{NN}^{-1} + \Lambda^{-1} \right)^{-1}, \;\;\; \mu^* = \Sigma^* \left( \Sigma_{NN}^{-1} \Sigma_{No} \Sigma_{2}^{-1} \tilde{\mu} + \Lambda^{-1} \tilde{\xi}_N \right),
\end{equation}
where \( \Sigma_2 = \Sigma_{oo} - \Sigma_{oN} \Sigma_{NN}^{-1} \Sigma_{No} \), \( \Lambda = \tau^2 (1 - \delta^2)^{-1} \operatorname{diag}\{ Z_{N,1}^{-1/2} \} \Upsilon \operatorname{diag}\{ Z_{N,1}^{-1/2} \} \), \\ \( \tilde{\xi}_N = \xi + \tau \delta \operatorname{diag}\{ Z_{N,1}^{-1/2} \} Z_{N,2} \), and \( \tilde{\mu} = Y_o - \mu + \Sigma_{oN} \Sigma_{NN}^{-1} \mu_N \).

The full conditional distribution of \( Z_{N,2} \) is a truncated multivariate normal, restricted to \( (\mathds{R}^+)^{|N|} \), with mean vector \( \mu^* \) and covariance matrix \( \Sigma^* \) given by
\begin{align}
\Sigma^* =& \left( \delta^2 (1 - \delta^2) \operatorname{diag}\{ Z_{N,1}^{-1/2} \} \Upsilon^{-1} \operatorname{diag}\{ Z_{N,1}^{-1/2} \} + I_N \right)^{-1},\\
\mu^* =& \frac{\delta (1 - \delta^2)}{\tau} \Sigma^* \operatorname{diag}\{ Z_{N,1}^{-1/2} \} \Upsilon^{-1} (Y_N - \xi_N).
\end{align}

The block \( Z_{N,1} \) is sampled using a \( \Delta \)-guided mixed pCN kernel.

The proposal distribution \( q(\ddot{Z}_{N_k,j} | Y_{N_k,j}, N_k, \theta_f) \) is a Normal-Gamma distribution, factorized as \( q(\ddot{Z}_{N_k,2,j} | \ddot{Z}_{N_k,1,j}, Y_{N_k,j}, N_k, \theta_f) q(\ddot{Z}_{N_k,1,j} | Y_{N_k,j}, N_k, \theta_f) \), where
\begin{eqnarray}
&&(\ddot{Z}_{\ddot{N}_k,2,j} | \ddot{Z}_{\ddot{N}_k,1,j}, Y_{\ddot{N}_k,j}, \ddot{N}_k, \theta_f) \sim N\left( \frac{\delta (1 - \delta^2) \ddot{Z}_{\ddot{N}_k,1,j}^{1/2} \left( \frac{Y_{\ddot{N}_k,j} - \xi}{\sigma} \right)}{1 + \delta^2 (1 - \delta^2)}, \; c_1 \frac{1}{1 + \delta^2 (1 - \delta^2)} \right), \\
&&(\ddot{Z}_{\ddot{N}_k,1,j} | Y_{\ddot{N}_k,j}, \ddot{N}_k, \theta_f) \sim G\left( c_2 \frac{\nu + 1}{2}, \; c_2 \left( \frac{\nu}{2} + \frac{(1 - \delta^2)}{2} \left( 1 - \frac{\delta^2}{1 + \delta^2 (1 - \delta^2)} \right) \left( \frac{Y_{\ddot{N}_k,j} - \xi}{\sigma} \right)^2 \right) \right).
\end{eqnarray}

To ensure that the base GP and the distribution \( f \) have matching means and variances, we set \( \tau^2 = \varsigma^{-1} \sigma^2 \) and \( \xi = \mu - \eta \), where \( \varsigma = \left( \frac{\nu}{\nu - 2} (\delta^2 + (1 - \delta^2)^{-1}) - \frac{\delta^2 \nu}{\pi} \left( \frac{\Gamma((\nu - 1)/2)}{\Gamma(\nu/2)} \right)^2 \right) \) and \( \eta = \tau \delta \sqrt{\nu/\pi} \frac{\Gamma((\nu - 1)/2)}{\Gamma(\nu/2)} \). The distribution of $X_{N}\mid Z_N$ is obtained by making $\xi=- \eta$ and $\tau^2=\varsigma^{-1}$ in \eqref{YN_St}.

\item Copula distributions

Suppose we want each coordinate of $Y_N$ to have a marginal density $\frac{1}{\tau}h\left(\frac{Y_{N,j}-\xi}{\tau}\right)$, where \( h \) is a density function with mean 0, variance 1, and corresponding c.d.f. \( H \). To define a correlation structure on \( Y_N \) through a Gaussian copula, define \( X_N \sim N_d(0, \Upsilon) \), where \( \Upsilon \) is an \( N \times N \) correlation matrix. Then, the \( j \)-th coordinate of \( Y_N \) is given by
\begin{equation}
Y_{N,j} = g^{-1}(X_{N,j}) = \xi + \tau H^{-1}(\Phi(X_{N,j})), \label{YN_CP}
\end{equation}
where \( \Phi \) denotes the standard normal c.d.f.

The \( j \)-th coordinate of \( g(Y_N) \) can be expressed as
\begin{equation}
X_{N,j} = g(Y_{N,j}) = \Phi^{-1}\left( H\left( \frac{Y_{N,j} - \xi}{\sigma} \right) \right).
\end{equation}

We sample the block \( g(Y_N) \) via pCN Metropolis-Hastings.

To ensure that the base GP and the distribution \( f \) have matching means and variances, we set \( \tau^2 = \sigma^2 \) and \( \xi = \mu \).

\end{enumerate}


\section{Non-Gaussian processes from \citet{Genton12}}\label{appendixD}

\citet{Genton12} propose the following processes with skew-normal, Student-t, and skew-t finite-dimensional distributions, respectively:
\begin{eqnarray}
Y_s &=& \mu + \sigma\left(\delta Z_{2,s} + (1 - \delta^2)^{-1/2} X_s\right), \label{NGex1} \\
Y_s &=& \mu + \sigma Z_{1,s}^{-1/2} X_s, \label{NGex2} \\
Y_s &=& \mu + \sigma Z_{1,s}^{-1/2} \left(\delta Z_{2,s} + (1 - \delta^2)^{-1/2} X_s\right), \label{NGex3}
\end{eqnarray}
where \(Z_2 = \{Z_{2,s}; s \in S\}\) are i.i.d. half-normal random variables, \(Z_1 = \{Z_{1,s}; s \in S\}\) are i.i.d. \(\text{Gamma}(\nu/2, \nu/2)\) random variables, and \(X\) is a Gaussian process.

Defining these processes with i.i.d. random fields resolves identifiability issues for parameters \(\delta\) and \(\nu\). However, these random fields also introduce the measurability and rough path problems discussed in Section~\ref{ACM}.

Figure~\ref{sim2} shows one realization from each of the three processes, with parameters \(\mu = 0\), \(\sigma^2 = 1\), \(\nu = 5\), and \(\delta = 0.95\), evaluated on a grid with 160,000 points. The Gaussian process \(X\) is simulated using Vecchia approximation, with a power exponential correlation function, allowing for simulation over such a large grid.

For the Student-t and skew-t processes, the plots are truncated to emphasize the rough paths. The actual minimum and maximum values across the mesh are \(-5.3\) and \(7.1\) for the Student-t process, and \(-18.5\) and \(25.2\) for the skew-t process. Values outside the displayed range are shown in white.

\begin{figure}[!h]
\centering
\includegraphics[width=\linewidth]{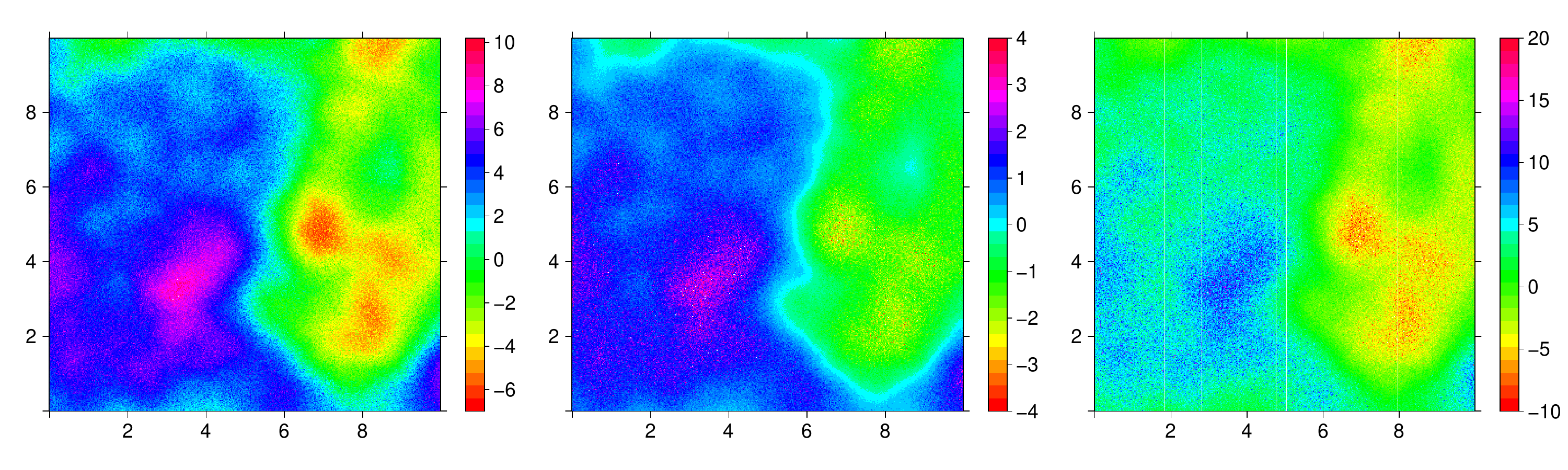}
\caption{Heat maps of a single realization from the skew-normal (left), Student-t (center), and skew-t (right) processes, evaluated on a grid of 160,000 locations.}
\label{sim2}
\end{figure}


\section{Updating the inverse of a positive-definite matrix}\label{appendixC}

Suppose we have an \( n \times n \) covariance matrix \( \Sigma_n \) and its inverse \( \Sigma_n^{-1} \).

If \( k \) new locations are added as the last entries in the vector of locations, the inverse of the resulting \((n+k) \times (n+k)\) covariance matrix can be updated using the Schur complement. Specifically, let
\[
\Sigma_{n+k} = \begin{bmatrix}
    \Sigma_n & B \\
    B^\intercal & D \\
\end{bmatrix},
\]
where \( B \) is the \( n \times k \) cross-covariance matrix between the existing \( n \) locations and the new \( k \) locations. Then, by the Schur complement, we have
\[
\Sigma_{n+k}^{-1} = \begin{bmatrix}
    \Sigma_n^{-1} + \Sigma_n^{-1} B \Upsilon B^\intercal \Sigma_n^{-1} & -\Sigma_n^{-1} B \Upsilon \\
    -\Upsilon B^\intercal \Sigma_n^{-1} & \Upsilon \\
\end{bmatrix},
\]
where \( \Upsilon = (D - B^\intercal \Sigma_n^{-1} B)^{-1} \) is a \( k \times k \) matrix.

In the case where \( k \) locations are removed, we first move those \( k \) locations to the last \( k \) positions in the vector of locations. This requires reordering the inverse covariance matrix so that the rows and columns corresponding to the removed locations are the last \( k \) rows and columns. Let the reordered inverse matrix be
\[
\Sigma_n^{-1} = \begin{bmatrix}
    A & B \\
    C & D \\
\end{bmatrix},
\]
where \( D \) is \( k \times k \). Then, the inverse of the reduced \((n - k) \times (n - k)\) covariance matrix is given by
\begin{equation}
    \Sigma_{n-k}^{-1} = A - B D^{-1} C.
\end{equation}

The following proposition establishes this result.

\begin{prop}
Consider an \( n \times n \) positive-definite matrix \( M = \begin{bmatrix} A & B \\ C & D \end{bmatrix} \), where \( D \) is \( k \times k \), and its inverse \( M^{-1} = \begin{bmatrix} a & b \\ c & d \end{bmatrix} \). Then,
\[
A^{-1} = a - b\,d^{-1}\,c.
\]
\end{prop}

\begin{proof}
By the Schur complement, we have \( A = (a - b d^{-1} c)^{-1} \).
\end{proof}

\end{document}